\documentclass{article}
\usepackage{graphicx}
\usepackage{float}
\usepackage{nameref}
\usepackage{amsthm}
\usepackage{amsmath}
\usepackage{amsfonts}
\usepackage{xassoccnt}
\usepackage{hyperref}
\hypersetup{colorlinks, linkcolor = {red}, citecolor = {green}, urlcolor = {blue}}
\usepackage{cleveref}
\usepackage{mathrsfs}
\usepackage{stmaryrd}
\usepackage{amssymb}
\usepackage[T1]{fontenc}
\usepackage{pxfonts}
\DeclareMathAlphabet{\mathpzc}{OMS}{cmsy}{m}{n}
\usepackage{tikz}
\usepackage{tikz-cd}
\usetikzlibrary{automata, arrows.meta, positioning}
\usepackage[english]{babel}
\usepackage{lineno}
\usepackage[margin=1.1in]{geometry}

\newtheorem{theorem}{Theorem}[section]
\crefname{theorem}{theorem}{theorems}
\newtheorem{claim}{Claim}[section]
\crefname{claim}{claim}{claims}
\newtheorem{proposition}{Proposition}[section]
\crefname{proposition}{proposition}{propositions}
\newtheorem{lemma}{Lemma}[section]
\crefname{lemma}{lemma}{lemmas}
\newtheorem{corollary}{Corollary}[section]
\crefname{corollary}{Corollary}{corollaries}
\newtheorem{conjecture}{Conjecture}[section]
\crefname{conjecture}{Conjecture}{conjectures}
\newtheorem{remark}{Remark}[section]
\crefname{remark}{remark}{remarks}
\newtheorem{example}{Example}[section]
\crefname{example}{example}{examples}
\newtheorem{definition}{Definition}[section]
\crefname{definition}{definition}{definitions}

\DeclareCoupledCountersGroup{theorems}
\DeclareCoupledCounters[name=theorems]{theorem,claim,proposition,lemma,remark,example,corollary,definition,conjecture}

\title{The complexity of smooth words over binary alphabets}
\author{Julien Cassaigne\footnote{CNRS, Aix Marseille Univ, I2M, 163 Avenue de Luminy, Case 907, 13288 Marseille Cedex 9, France} ~~~~~~~~~~~~ Raphaël Henry\footnote{Aix Marseille Univ, CNRS, I2M, 3 place Victor Hugo, Case 19, 13331 Marseille Cedex 3, France}}
\date{}

\begin{document}

\maketitle

\pagenumbering{arabic}
\setcounter{page}{1}

\begin{abstract}
	Smooth sequences over an alphabet of positive integers $\{a,b\}$ are sequences that are infinitely derivable, the emblematic example being the Oldenburger-Kolakoski sequence $\kappa$ over $\{1,2\}$. The main way to study their language is to consider a finite version of smooth sequences called smooth words. In this paper we prove that the smooth words are exactly the factors of smooth sequences and we make progress towards the conjecture that the complexity of smooth words over $\{a,b\}$ grows like $\Theta\left(n^{\log(a+b)/\log((a+b)/2)}\right)$: we prove the lower bound over any alphabet and the upper bound over even alphabets, and we improve the known upper bound over odd and mixed alphabets. We also provide a dichotomy that links the language, the morphicity and the factor complexity of $\kappa$.
\end{abstract}

\newcommand{\kola}{
    \begin{gathered}
    \tikzpicture[every node/.style={anchor=south west}]
        \node[minimum width=1cm,minimum height=0.4cm] at (0,0) {$\kappa = \underbrace{22}_2 \underbrace{11}_2 \underbrace{2}_1 \underbrace{1}_1 \underbrace{22}_2 \underbrace{1}_1 \underbrace{22}_2 \underbrace{11}_2 \underbrace{2}_1 \underbrace{11}_2 \underbrace{22}_2 \underbrace{1}_1 \underbrace{2}_1 \underbrace{22}_2 \underbrace{1}_1 ...$};
        \node[minimum width=1cm,minimum height=0.4cm] at (14.15,0.05) {$... = \kappa$};
    \endtikzpicture
    \end{gathered}
}

\newcommand{\kolabis}{
    \begin{gathered}
    \tikzpicture[every node/.style={anchor=south west}]
        \node[minimum width=1cm,minimum height=0.4cm] at (0,0) {$\kappa_{3,1} = \underbrace{333}_3 \underbrace{111}_3 \underbrace{333}_3 \underbrace{1}_1 \underbrace{3}_1 \underbrace{1}_1 \underbrace{333}_3 \underbrace{111}_3 \underbrace{333}_3 \underbrace{1}_1 \underbrace{333}_3 \underbrace{1}_1 \underbrace{333}_3 \underbrace{111}_3 ...$};
        \node[minimum width=1cm,minimum height=0.4cm] at (13.575,-0.05) {$... = \kappa_{3,1}$};
    \endtikzpicture
    \end{gathered}
}

\section*{Introduction}

The emblematic Oldenburger-Kolakoski sequence is a fixed point of the \textit{run-length encoding} over the alphabet $\{1,2\}$:
$$\kola$$
Initially studied by Oldenburger \cite{Oldenburger}, it was later popularized by Kolakoski \cite{Kolakoski} and is now registered as A000002 in the OEIS \cite{OEIS}. Early on, the community of combinatorics on words asked the following questions: Does $\kappa$ belong to a known class of sequences? Is it eventually periodic? (uniformly) recurrent? Does it have factor frequencies? What is its factor complexity? Kolakoski's original question had already been answered in \cite{Oldenburger}: $\kappa$ is not eventually periodic. Also, in \cite{CKL} the authors show that $\kappa$ is a PD0L sequence but not purely morphic. The other questions remain unsolved.

One of the main challenges is to describe the factors of $\kappa$. To do so, Dekking introduced in \cite{Dekking81} \textit{smooth words} (also called $C^\infty$-words), that is, a set of finite words over $\{1,2\}$ closed under an operation called \textit{(finite) derivative} that simulates locally the run-length encoding. In this paper we denote the set of smooth words by $\mathcal{C}_f^\infty$. It easily follows that all factors of $\kappa$ are smooth words, and Dekking conjectured the converse, which has led to a thorough investigation of the properties of smooth words. Notably, Carpi showed in \cite{Carpi} that smooth words are cube-free, which implies that $\kappa$ is cube-free. Also, Dekking conjectured in \cite{Dekking81} that the complexity of smooth words grows like $\Theta\left(n^{\log(3)/\log(3/2)}\right)$ and was able to show that it is bounded by two polynomials. Weakley was the first to study the bispecial smooth words, in \cite{Weakley} he described how they can be used to prove the conjecture on the complexity and in \cite{HW} he and Huang obtained other polynomial bounds of the complexity that rely on letter frequencies.

Facing the difficulties raised by $\kappa$, researchers started to investigate similar sequences in the hope of finding better answers. The first idea is to change the alphabet: if $a$ and $b$ are two different integers, let $\kappa_{a,b}$ denote the infinite fixed point of the run-length encoding starting with $a$ over the binary alphabet $\{a,b\}$. With this notation, the Oldenburger-Kolakoski sequence is denoted by $\kappa_{2,1}$ and we observe that $\kappa_{1,2} = 1 \kappa_{2,1}$. The next sequence to bring attention, registered as A064353 in the OEIS \cite{OEIS}, is
$$\kolabis$$
Dekking had already noticed in \cite{Dekking79} that $\kappa_{3,1}$ is morphic, and in \cite{BS} the authors extensively studied its properties. In particular it is linearly recurrent, it has linear complexity and it has algebraic letter frequencies that are different from $1/2$ (the frequency of $3$'s is approximately 0.6028). Interestingly, $\kappa_{3,1}$ is not the only sequence with these properties: in the overview \cite{Sing10}, Sing showed that all sequences $\kappa_{a,b}$ are generated by a primitive substitution when $a+b$ is even.

In this paper we shall consider a family of sequences originally introduced over $\{1,2\}$ by Dekking \cite[Section 4]{Dekking01}, which contains the sequences $\kappa_{a,b}$. Over any binary alphabet $\mathpzc{A} = \{a,b\}$, if the run-length encoding is \textit{derivative}, then a sequence in $\mathpzc{A}^\mathbb{N}$ is \textit{smooth} if all its derivatives are in $\mathpzc{A}^\mathbb{N}$. We write $\mathcal{C}^\infty$ the set of smooth sequences, and we define smooth words and their set $\mathcal{C}_f^\infty$ over any binary alphabet as we did over $\{1,2\}$. We then observe a dichotomy between smooth sequences:
\begin{itemize}
	\item Over alphabets $\mathpzc{A} = \{a,b\}$ where $a+b$ is odd, it seems that every smooth sequence contains all smooth words. In particular, all smooth sequences have the same factor complexity as $\mathcal{C}_f^\infty$.
	\item Over alphabets $\mathpzc{A} = \{a,b\}$ where $a+b$ is even, some smooth sequences have linear factor complexity and do not contain all smooth words.
\end{itemize}
However, smooth words are still very relevant in the second case of the dichotomy because 
the factors of smooth sequences are always smooth. Therefore the attempt to describe smooth sequences by studying smooth words naturally developed, and researchers started to generalize the properties of smooth words from $\{1,2\}$ to every binary alphabet. For instance Huang \cite{powerFree} generalized the result of \cite{Carpi} by showing that smooth words are $(b+1)$-th-power-free over any alphabet $\{a,b\}$ where $a < b$, except over $\{1,3\}$ where they are $5$-th-power-free (we must point out that it is an unpublished paper). Regarding the complexity, Sing \cite{Sing02} extended the conjecture of Dekking by stating that ${p_{\mathcal{C}_f^\infty}(n) = \Theta\left(n^{\log(a+b)/\log((a+b)/2)}\right)}$ over any alphabet $\{a,b\}$ and generalized Dekking's polynomial bounds. Also, Huang claimed in \cite{Huang} to generalize the results of \cite{HW} to any binary alphabet, but we found a mistake in these results that we explain in the preliminaries of this paper.

In this paper we contribute to the study of smooth words.
\begin{itemize}
	\item With \Cref{Th1} we show that, over any binary alphabet, the smooth words are exactly the factors of smooth sequences. This yields the equality of complexities ${p_{\mathcal{C}^\infty}(n) = p_{\mathcal{C}_f^\infty}(n)}$ and further motivates the study of $\mathcal{C}_f^\infty$.
	
	\item With \Cref{Th2} we prove the conjectured lower bound of $p_{\mathcal{C}_f^\infty}(n)$ over any alphabet.
	
	\item With \Cref{Th3} we obtain three upper bounds of $p_{\mathcal{C}_f^\infty}(n)$ based on the same argument: an unrefined upper bound over any alphabet $\{a,b\}$, the conjectured upper bound over \textit{even} alphabets (i.e., $a$ and $b$ are even) and a new upper bound over \textit{odd} alphabets (i.e., $a$ and $b$ are odd).
	
	\item With \Cref{Th4} we obtain upper bounds of $p_{\mathcal{C}_f^\infty}(n)$ under several assumptions on letter frequencies in $\mathcal{C}_f^\infty$. In particular we provide a necessary condition to prove the conjectured upper bound.
\end{itemize}

Lastly we come back to the alphabet $\{1,2\}$ and to the original Oldenburger-Kolakoski sequence $\kappa$.

\begin{itemize}
	\item With \Cref{dichotomy} we show a dichotomy that links the conjecture that $\mathscr{L}(\kappa) = \mathcal{C}_f^\infty$, the morphicity and the factor complexity of $\kappa$.
\end{itemize}

To conclude, let us discuss two blind spots this paper has regarding smooth sequences. The first one is that, according to the dichotomy on smooth sequences, studying the language $\mathcal{C}_f^\infty$ only helps describing the (expected) language of smooth sequences over mixed alphabets. For all we know, the language of smooth sequences over even and odd alphabets has been unexplored beyond the fixed points $\kappa_{a,b}$. The second one is that, although it appears in \Cref{Th4}, we do not tackle the important question of factor frequencies in smooth sequences which generalizes the famous conjecture of Keane \cite{Keane} that $\kappa$ has equal letter frequencies. Let us cite two contributions to this question: in \cite{BJP} the authors computed the frequency of letters in the lexicographic minimal and maximal smooth sequences; and in \cite{JMPR} the authors proved the existence of frequencies and characterized the letter frequencies of smooth sequences over $\{1,3\}$.

\subsubsection*{Outline of the paper}

In \Cref{prelim} we provide the basic definitions of our paper and we define smooth sequences and smooth words over binary alphabets. Then we recall the previous results on the factor complexity of smooth words, and we take the time to explain the mistake in the results of \cite{Huang}. Finally we formulate and reflect on our results.
\newline
In \Cref{factors} we prove \Cref{Th1}.
\newline
In \Cref{bispecials} we describe the bispecial smooth words.
\newline
In \Cref{bounds} we prove \Cref{Th2,Th3,Th4}.
\newline
In \Cref{consequences} we apply our results over the alphabet $\{1,2\}$ and we prove \Cref{dichotomy}.

\section{Preliminaries}\label{prelim}

\subsection{Combinatorics on words}\label{CoW}

\subsubsection{Words basics}

\begin{definition}
	A binary alphabet is a set $\mathpzc{A} = \{a,b\}$ where $a$ and $b$ are integers such that $1 \leq a < b$. If $a$ and $b$ are both even (resp. both odd), we say that $\mathpzc{A}$ is \textup{even} (resp. \textup{odd}). Otherwise, i.e., if $a+b$ is odd, we say that $\mathpzc{A}$ is \textup{mixed}.
\end{definition}

Given a finite alphabet $\mathpzc{A}$, $\mathpzc{A}^n$ (resp. $\mathpzc{A}^*$) denotes the set of finite words of length $n$ (resp. the set of all finite words) over $\mathpzc{A}$, and $\lvert u \rvert$ denotes the length of the word $u \in \mathpzc{A}^*$. In particular $\varepsilon$ is the empty word and we define $\mathpzc{A}^+ := \mathpzc{A}^*\backslash\{\varepsilon\}$. If words $u,v,w,w' \in \mathpzc{A}^*$ are such that $v = wuw'$, we say that $u$ is a \textit{factor} of $v$ and we write $u \sqsubset v$. Moreover, if $w = \varepsilon$, we say that $u$ is a \textit{prefix} of $v$ and we write $u \sqsubset_p v$. Over a binary alphabet $\mathpzc{A} = \{a,b\}$, the \textit{complement} of a word $u \in \mathpzc{A}^*$ (resp. $x \in \mathpzc{A}^\mathbb{N}$) is the word $\overline{u}$ (resp. $\overline{x}$) where every letter $a$ is replaced with the letter $b$ and vice versa. If $u = u_1 ... u_n \in \mathpzc{A}^n$ and $c \in \mathpzc{A}$, we define $\lvert u \rvert_c := \left\lvert \left\{ i \in \left\llbracket 1,n \right\rrbracket ~\middle|~ u_i = c \right\} \right\rvert$.

Let $\mathpzc{A}^\mathbb{N}$ be the set of all (right-infinite) sequences over $\mathpzc{A}$, where $0 \in \mathbb{N}$. The factors of a sequence $x \in \mathpzc{A}^\mathbb{N}$ are the finite words $u$ of the form $x_{[i,j)} := x_i x_{i+1} ... x_{j-1}$ for $0 \leq i \leq j$, with the convention that $x_{[i,i)} := \varepsilon$, and we write $u \sqsubset x$. If $u$ is of the form $x_{[0,j)}$, we say that $u$ is a \textit{prefix} of $x$ and we write $u \sqsubset_p x$. The \textit{language} of a sequence $x \in \mathpzc{A}^\mathbb{N}$ is the set of its factors ${\mathscr{L}(x) := \{u \in \mathpzc{A}^* \mid u \sqsubset x\}}$. Similarly, the language of a set $X \subseteq \mathpzc{A}^\mathbb{N}$ is the set $\mathscr{L}(X) := \displaystyle \cup_{x \in X} \mathscr{L}(x)$.

A sequence $x \in \mathpzc{A}^\mathbb{N}$ is \textit{recurrent} if each of its factors occurs infinitely often in $x$. A word $u \in \mathpzc{A}^*$ \textit{occurs with bounded gaps} in $x$ if there exists $L \geq 1$ such that, for all $v \in \mathscr{L}(x)$, if $\lvert v \rvert \geq L$ then $u \sqsubset v$. Also, stronger than being recurrent, $x$ is \textit{uniformly recurrent} if all its factors occur with bounded gaps in it. Even stronger than being uniformly recurrent, $x$ is \textit{linearly recurrent} if there exists a constant $L \geq 1$ such that, for all $u,v \in \mathscr{L}(x)$, if $\lvert v \rvert \geq L \lvert u \rvert$ then $u \sqsubset v$.
\newpage
\subsubsection{Factor complexity and bispecial words}\label{complexity}

\begin{definition}
	We define the following complexity functions $p : \mathbb{N} \rightarrow \mathbb{N}$.
	\begin{itemize}
		\item The \textup{complexity} of a language $\mathscr{L} \subseteq \mathpzc{A}^*$ is
	${p_\mathscr{L}(n) := \left\lvert \mathscr{L} \cap \mathpzc{A}^n \right\rvert}$,
		\item The \textup{factor complexity} of a sequence $x \in \mathpzc{A}^\mathbb{N}$ is $p_x(n) := p_{\mathscr{L}(x)}(n)$,
		\item The \textup{factor complexity} of a set of sequences $X \subseteq \mathpzc{A}^\mathbb{N}$ is $p_X(n) := p_{\mathscr{L}(X)}(n)$.
	\end{itemize}
\end{definition}

We use the following notations for the asymptotics of a complexity function.

\begin{definition}
	Given two functions $p,f : \mathbb{N} \rightarrow \mathbb{N}$, we write
	\begin{itemize}
		\item $p(n) = \mathcal{O}(f(n))$ if there exists $C > 0$ and $N \geq 0$ such that ${p(n) \leq C f(n)}$ for all $n \geq N$,
		\item $p(n) = \Omega(f(n))$ if there exists $C > 0$ and $N \geq 0$ such that ${p(n) \geq C f(n)}$ for all $n \geq N$,
		\item ${p(n) = \Theta(f(n))}$ if ${p(n) = \mathcal{O}(f(n))}$ and ${p(n) = \Omega(f(n))}$.
	\end{itemize}
\end{definition}

Let us recall how to compute the complexity of a language over a binary alphabet from the bispecial factors (we refer to \cite[Section 3.3]{Cassaigne}). First, a language $\mathscr{L} \subseteq \mathpzc{A}^*$ is \textit{factorial} if $u \sqsubset v \in \mathscr{L}$ implies $u \in \mathscr{L}$, and it is \textit{left-extendable} (resp. \textit{right-extendable}) if, for all $u \in \mathscr{L}$, there exists $a \in \mathpzc{A}$ such that $au \in \mathscr{L}$ (resp. $ua \in \mathscr{L}$).

\begin{definition}
	Let $\mathpzc{A} = \{a,b\}$ be a binary alphabet and let $\mathscr{L} \subseteq \mathpzc{A}^*$ be a factorial, left- and right-extendable language. A word $u \in \mathscr{L}$ is \textup{bispecial} if ${\{au,bu,ua,ub\} \subseteq \mathscr{L}}$. In that case, the set $\{aua,aub,bua,bub\} \cap \mathscr{L}$ has 2, 3 or 4 elements and its cardinal determines the \textup{type} of $u$: we say that $u$ is \textup{weak} if it has 2 elements, \textup{neutral} if it has 3 and \textup{strong} if it has 4.
\end{definition}

Let $s(n)$ be the first finite difference of the factor complexity, i.e., $s(n) := p_\mathscr{L}(n+1) - p_\mathscr{L}(n)$, and let $b(n)$ be the second finite difference, i.e., $b(n) := s(n+1) - s(n)$. The seminal result that allows to compute the complexity of the language $\mathscr{L}$ is the following.

\begin{theorem}[{\cite[Proposition 3.2]{Cassaigne}}]\label{pfrombsw}
	Let $\mathscr{L}$ be a factorial, left- and right-extendable language over a binary alphabet, and let $bs(n)$ (resp. $bw(n)$) denote the number of strong (resp. weak) bispecial words of $\mathscr{L}$. Then for all $n \geq 0$ we have
	$$b(n) = bs(n) - bw(n).$$
\end{theorem}

\subsection{Smooth sequences and smooth words}

In \cite{Dekking81}, Dekking initiated the investigation of the factor complexity of the Oldenburger-Kolakoski sequence $\kappa$. His strategy was to introduce the computable set of \textit{smooth words}, which are expected to be exactly the factors of $\kappa$. In this subsection we introduce \textit{smooth sequences} as a generalization of the sequence $\kappa$ and the \textit{smooth words} over any binary alphabet. We then detail how smooth words are related to smooth sequences and we state the conjectures and results about the complexity of smooth words.

\subsubsection{Smooth sequences}

Over a binary alphabet $\mathpzc{A} = \{a,b\}$, any sequence $x \in \mathpzc{A}^\mathbb{N}$ where the letters $a$ and $b$ occur infinitely often is canonically factorized as ${x_0}^{p_0} \overline{x_0}^{p_1} {x_0}^{p_2} \overline{x_0}^{p_3} ...$ where $x_0$ is the first letter of $x$ and $p_i \geq 1$ for all $i \geq 0$.

\begin{definition}
	Over a binary alphabet $\mathpzc{A}$, we say that a sequence $x \in \mathpzc{A}^\mathbb{N}$ is \textup{derivable} if its canonical factorization is ${x_0}^{p_0} \overline{x_0}^{p_1} {x_0}^{p_2} \overline{x_0}^{p_3}...$ with $p_i \in \mathpzc{A}$ for all $i \geq 0$. We write $\mathcal{C}^1$ the set of derivable sequences and we define the \textup{derivative} as the map
	\begin{align*}
		\mathcal{D}~:~~~~~~~~~~~~~~~~~~ \mathcal{C}^1 & \longrightarrow \mathpzc{A}^\mathbb{N} \\
    	{x_0}^{p_0} \overline{x_0}^{p_1} {x_0}^{p_2} ... & \longmapsto p_0 p_1 p_2...
	\end{align*}
\end{definition}

\begin{example}
	Over $\{1,2\}$, we have $(221)^{\omega} \in \mathcal{C}^1$ and $\mathcal{D}\left((221)^{\omega}\right) = (21)^{\omega} \in \mathcal{C}^1$ but $\mathcal{D}\left((21)^{\omega}\right) = 1^{\omega} \notin \mathcal{C}^1$.
\end{example}

\begin{definition}
	Over a binary alphabet $\mathpzc{A}$, we define by induction the set of sequences that are derivable $n$ times for $n \geq 2$ by
	\begin{equation*}
		\mathcal{C}^n := \left\{x \in \mathcal{C}^1 ~\middle |~ \mathcal{D}(x) \in \mathcal{C}^{n-1} \right\}.
	\end{equation*}
	We then define the set of \textup{smooth sequences} over $\mathpzc{A}$, that is the set of sequences that are infinitely derivable
	\begin{equation*}
		\mathcal{C}^\infty := \displaystyle \bigcap_{n \geq 1} \mathcal{C}^n.
	\end{equation*}
\end{definition}

\begin{example}\label{kab}
	As a primary example of smooth sequences, over any binary alphabet $\mathpzc{A} = \{a,b\}$, the derivative $\mathcal{D}$ has two fixed points: one that starts with $a$, denoted by $\kappa_{a,b}$; the other that starts with $b$, denoted by $\kappa_{b,a}$. Here are notable examples, among which is the Oldenburger-Kolakoski sequence $\kappa_{2,1}$:
	\begin{align*}
		\kappa_{2,1} = &~2211212212211211221211212211211212212211212212112112212211212212211211212212112... \\
		\kappa_{3,1} = &~3331113331313331113331333133311133313133311133313331113331333111333131333111333... \\
		\kappa_{2,4} = &~2244222244442244224422224444222244442244222244442244222244442244224422224444222... \\
		\kappa_{2,5} = &~2255222225555522552255225555522222555552222255555225522222555552255222225555522...
	\end{align*}
	Also, we notice that $\kappa_{1,b} = 1 \kappa_{b,1}$ for all $b \geq 2$.
\end{example}

Oldenburger originally proved in \cite[Theorem 2]{Oldenburger} that $\kappa_{2,1}$ is not eventually periodic. It was proven again by Üçoluk in \cite{Ucoluk} and by Dekking in \cite[Example 4]{Dekking79}, and the proof is easily adaptable to any smooth sequence.

\begin{proposition}\label{aperiodic}
	No smooth sequence is eventually periodic.
\end{proposition}

\subsubsection{Smooth words}

Over a binary alphabet $\mathpzc{A}$, any word $u \in \mathpzc{A}^*$ is canonically factorized as $a_1^{p_1} ... a_n^{p_n}$ where $a_i \in \mathpzc{A}$, $a_{i+1} = \overline{a_i}$ and $p_i \geq 1$, with the convention that $n = 0$ if $u = \varepsilon$.

\begin{definition}\label{Df}
	Over a binary alphabet $\mathpzc{A}$, we say that a word $u \in \mathpzc{A}^*$ is \textup{derivable} if its canonical factorization is $a_1^{p_1} ... a_n^{p_n}$ with $p_i \in \mathpzc{A}$ for all $i \in \llbracket 2,n-1 \rrbracket$ and ${p_1,p_n \in \llbracket 1,b \rrbracket}$. We write $\mathcal{C}^1_f$ the set of derivable words and we define the \textup{finite derivative} as the map
	\begin{align*}
		\mathcal{D}_f~:~~~~~~~~~ \mathcal{C}^1_f & \longrightarrow \mathpzc{A}^* \\
    	a_1^{p_1}...a_n^{p_n} & \longmapsto \textup{cut}(p_1)~p_2~...~p_{n-1}~\textup{cut}(p_n)
	\end{align*}
	where $\textup{cut}(p) := \begin{cases} \varepsilon \textup{ if } 1 \leq p \leq a \\ b \textup{ if } a < p \leq b \end{cases}$, $\mathcal{D}_f(\varepsilon) = \varepsilon$ and ${\mathcal{D}_f(a_1^{p_1}) = \textup{cut}(p_1)}$.
\end{definition}

\begin{example}
	Over $\{1,2\}$, $\mathcal{C}^1_f$ is the set of words that do not have $111$ or $222$ as factor, and we have ${\mathcal{D}_f(\varepsilon) = \mathcal{D}_f(1) = \mathcal{D}_f(21) = \varepsilon}$ and $\mathcal{D}_f(2211) = \mathcal{D}_f(122112) = 22$. Over $\{1,3\}$, we have $\mathcal{D}_f(331113) = 33$.
\end{example}

\begin{definition}
	For $n \geq 2$, we define by induction the set of words that are derivable $n$ times
	\begin{equation*}
		\mathcal{C}_f^n := \left\{u \in \mathcal{C}^1_f ~\middle |~ \mathcal{D}_f(u) \in \mathcal{C}_f^{n-1} \right\}.
	\end{equation*}
	We then define the set of \textup{smooth words}, that is the set of words that are infinitely derivable
	\begin{equation*}
		\mathcal{C}_f^\infty := \displaystyle \bigcap_{n \geq 1} \mathcal{C}_f^n.
	\end{equation*}
\end{definition}

We notice that the map $\mathcal{D}_f$ is \textit{contracting}, i.e., $\left\lvert \mathcal{D}_f(u) \right\rvert < \lvert u \rvert$ if ${u \in \mathcal{C}^1_f \backslash \{\varepsilon\}}$. This implies that a word is smooth if and only if it can be derived down to $\varepsilon$.

\begin{definition}
	The \textup{height} of a smooth word $u$ is the first integer $n \geq 0$ such that $\left(\mathcal{D}_f\right)^n(u) = \varepsilon$.
\end{definition}

\begin{example}
	Over $\{1,2\}$, we have $\left(\mathcal{D}_f\right)^4(221121221) = \left(\mathcal{D}_f\right)^3(22112) = \left(\mathcal{D}_f\right)^2(22) = \mathcal{D}_f(2) = \varepsilon$ so $221121221 \in \mathcal{C}_f^\infty$ with height $4$, but $\mathcal{D}_f(12121) = 111 \notin \mathcal{C}^1_f$ so $12121 \notin \mathcal{C}_f^\infty$.
\end{example}

One can easily show that the language $\mathcal{C}_f^\infty$ is factorial, left- and right-extendable, we refer to \cite[Proposition 2]{Weakley} for the proof over $\{1,2\}$.

\subsubsection{The language of smooth sequences}

From the definitions of smooth sequences and smooth words, it is clear that $\mathscr{L}(\mathcal{C}^\infty) \subseteq \mathcal{C}_f^\infty$, which yields $p_x(n) \leq p_{\mathcal{C}^\infty}(n) \leq p_{\mathcal{C}_f^\infty}(n)$ for all $x \in \mathcal{C}^\infty$. In \cite{Dekking81}, Dekking conjectured that $\mathscr{L}(\kappa_{2,1}) = \mathcal{C}_f^\infty$ over $\{1,2\}$ and showed that it is equivalent to the complement-invariance of $\mathscr{L}(\kappa_{2,1})$. In \cite{Sing10} Sing conjectured the same for all $\kappa_{a,b}$ over mixed alphabets, and we extend it even further.

\begin{conjecture}\label{mixedCf}
	Over mixed alphabets, every smooth sequence $x$ satisfies $\mathscr{L}(x) = \mathcal{C}_f^\infty$, and in particular ${p_x(n) = p_{\mathcal{C}^\infty}(n) = p_{\mathcal{C}_f^\infty}(n)}$.
\end{conjecture}

This is linked to the recurrence of smooth sequences.

\begin{remark}\label{invariant}
	If \Cref{mixedCf} holds, then the language of every smooth sequence over a mixed alphabet is complement-invariant and reversal-invariant (where the reversal of a word $u_1...u_n$ is $u_n...u_1$). Furthermore, a known result holding for all sequences is that complement-invariance or reversal-invariance implies recurrence (see \cite[Proposition 1]{Dekking81} for complement-invariance).
\end{remark}

However, \Cref{mixedCf} cannot be extended to even or odd alphabets. Over even alphabets, in  \cite[Proposition 26]{BJP} the authors showed that the language of the minimal (resp. maximal) smooth sequence with respect to the lexicographic order is neither reversal-invariant nor complement-invariant, contrary to $\mathcal{C}_f^\infty$. Over odd alphabets, the most pathological example is the following.

\begin{example}\label{1bb1}
	Over the odd alphabet $\mathpzc{A} = \{1,b\}$, there exists a pair of smooth sequences that are the derivative of each other, and they do not contain the smooth word $bb$. For example, over $\{1,3\}$ they are
	\begin{align*}
	& 111311131311131113131113131113111313111311131311131311131113131113131113111313111311... \\
	& 313111313111311131311131311131113131113111313111313111311131311131311131113131113111...
	\end{align*}
\end{example}

Even though \Cref{mixedCf} does not hold over even or odd alphabets, we will show in \Cref{Th1} that $\mathscr{L}(\mathcal{C}^\infty) = \mathcal{C}_f^\infty$, which yields $p_{\mathcal{C}^\infty}(n) = p_{\mathcal{C}_f^\infty}(n)$ and cements the importance of the language $\mathcal{C}_f^\infty$ over every binary alphabet.

\subsubsection{The complexity of smooth words}

With the previous observations and conjectures, the complexity $p_{\mathcal{C}_f^\infty}(n)$ has naturally been thoroughly investigated. The main conjecture was stated over $\{1,2\}$ by Dekking \cite{Dekking81} and generalized to any binary alphabet by Sing \cite{Sing02}.

\begin{conjecture}[{\cite{Sing02}}]\label{CfTheta}
	Over a binary alphabet $\mathpzc{A} = \{a,b\}$, we have $p_{\mathcal{C}_f^\infty}(n) = \Theta(n^\rho)$ where $\rho := \frac{\log(a+b)}{\log\left(\frac{a+b}{2}\right)}$.
\end{conjecture}

This is linked to the morphicity of smooth sequences.

\begin{corollary}\label{notMorphic}
	If \Cref{mixedCf,CfTheta} hold, then no smooth sequence over a mixed alphabet is morphic.
\end{corollary}
\begin{proof}
	If $x$ is a smooth sequence over a mixed alphabet $\mathpzc{A} = \{a,b\}$, \Cref{mixedCf,CfTheta} yield $p_x(n) = \Theta(n^\rho)$ where $\rho := \frac{\log(a+b)}{\log\left(\frac{a+b}{2}\right)}$. However, Devyatov's result in \cite{Devyatov} states that the factor complexity of a morphic sequence can only grow like $\Theta(n^{1+1/k})$ for some $k \geq 1$ or like $\mathcal{O}(n \log n)$, and we have $\rho = 1+1/k$ if and only if $a+b = 2^{k+1}$, which cannot happen when $\mathpzc{A} = \{a,b\}$ is mixed.
\end{proof}

A first attempt to study $p_{\mathcal{C}_f^\infty}(n)$ over $\{1,2\}$ was made by Dekking \cite{Dekking81} and later generalized to every binary alphabet by Sing \cite{Sing02}.

\begin{theorem}[{\cite[Propositions 1 and 3]{Sing02}}]\label{firstPoly}
	Let $\mathpzc{A} = \{a,b\}$ be a binary alphabet.
	
	\textup{(i)} We have $p_{\mathcal{C}_f^\infty}(n) = \Omega(n^\alpha)$ where $\alpha := \frac{\log(a+b)}{\log\left(\frac{a^2+b^2}{a+b}\right)}$.
	
	\textup{(ii)} We have $p_{\mathcal{C}_f^\infty}(n) = \mathcal{O}(n^\beta)$ where $\beta := \frac{\log\left(2b^2\right)}{\log\left(\frac{2ab}{a+b}\right)}$.
\end{theorem}

These first polynomial bounds over any binary alphabet imply zero topological entropy for all smooth sequences, where the \textit{topological entropy} of a sequence $x \in \mathpzc{A}^\mathbb{N}$ is the finite quantity $\displaystyle\lim_{n \rightarrow \infty} \log(p_x(n))/n$.

\begin{corollary}
	Every smooth sequence has zero topological entropy.
\end{corollary}

In \cite{Weakley}, Weakley started to investigate the bispecial smooth words over $\{1,2\}$ in order to study $p_{\mathcal{C}_f^\infty}(n)$.

\begin{theorem}[{\cite[Corollary 9]{Weakley}}]\label{almost12}
	Over $\{1,2\}$, let $\rho := \log(3)/\log(3/2)$. Then there exist two constants $0 < C_1 < C_2$ such that
	\begin{equation*}
		C_1 n^\rho \leq p_{\mathcal{C}_f^\infty}(n) \leq C_2 n^\rho
	\end{equation*}
	for all $n,i \geq 0$ such that $L_i+1 \leq n \leq \ell_{i+1}+1$, where $\ell_i$ (resp. $L_i$) is the minimum (resp. maximum) length of bispecial smooth words of height $i$.
\end{theorem}

\Cref{almost12} would prove \Cref{CfTheta} over $\{1,2\}$ if we have $L_i \leq \ell_{i+1}$ for all $i \geq 0$. Computations suggest that it is true, but it remains a conjecture.

\begin{conjecture}[{\cite{Weakley}}]\label{Lili1}
	Over $\{1,2\}$, we have $L_i \leq \ell_{i+1}$ for all $i \geq 0$.
\end{conjecture}

In \cite{HW}, Weakley and Huang used the fact that the length of bispecial smooth words is controlled by their letter frequencies to obtain the following result.

\begin{theorem}[{\cite[Theorem 3]{HW}}]\label{freqCf}
	Over $\{1,2\}$, let $\phi \in \left(0,\frac{1}{2}\right)$ and let $N$ be a positive integer such that $\frac{\lvert u \rvert_2}{\lvert u \rvert} > \frac{1}{2} - \phi$ for all $u \in \mathcal{C}_f^\infty$ such that $\lvert u \rvert > N$.
	
	\textup{(i)} We have $p_{\mathcal{C}_f^\infty}(n) = \Omega(n^\delta)$ where $\delta := \frac{\log(3)}{\log\left(\frac{3}{2} + \phi + \frac{2}{N}\right)}$.
	
	\textup{(ii)} We have $p_{\mathcal{C}_f^\infty}(n) = \mathcal{O}(n^\gamma)$ where $\gamma := \frac{\log(3)}{\log\left(\frac{3}{2} - \phi\right)}$.
\end{theorem}

We expect $\phi$ to be arbitrarily close to $0$ in \Cref{freqCf}, which supports \Cref{CfTheta} over $\{1,2\}$.

\subsubsection{A mistake in a paper}\label{mistake}

In \cite{Huang}, Huang claims to generalize \Cref{freqCf} to any binary alphabet.

\begin{claim}[{\cite[Theorem 12]{Huang}}]\label{mistake1}
	Over a binary alphabet $\mathpzc{A} = \{a,b\}$, let $\phi \in \left(0, \frac{1}{2}\right)$ and let $N$ is a positive integer such that $\frac{\lvert u \rvert_b}{\lvert u \rvert} > \phi$ for all $u \in \mathcal{C}_f^\infty$ such that $\lvert u \rvert > N$.
	
	\textup{(i)} We have $p_{\mathcal{C}_f^\infty}(n) = \Omega(n^\delta)$ where $\delta := \frac{\log(2b-1)}{\log(1+(a+b-2)(1-\phi))}$.
	
	\textup{(ii)} We have $p_{\mathcal{C}_f^\infty}(n) = \mathcal{O}(n^\gamma)$ where $\gamma := \frac{\log(2b-1)}{\log(1+(a+b-2)\phi)}$.
\end{claim}

Over mixed alphabets we expect $\phi$ to be arbitrarily close to $1/2$ in \Cref{mistake1}, which would suggest that $p_{\mathcal{C}_f^\infty}(n) = \Theta(n^{\rho'})$ where ${\rho' = \frac{\log(2b-1)}{\log(\frac{a+b}{2})}}$. In the same paper, Huang also claimed to give the asymptotic behavior of $p_{\mathcal{C}_f^\infty}(n)$ over even alphabets.
\newpage
\begin{claim}[{\cite[Theorem 16]{Huang}}]\label{mistake2}
	Let $\mathpzc{A} = \{a,b\}$ be an even alphabet. Then we have $p_{\mathcal{C}_f^\infty}(n) = \Theta(n^{\rho'})$ where $\rho' := \frac{\log(2b-1)}{\log\left(\frac{a+b}{2}\right)}$.
\end{claim}

We notice that \Cref{mistake1,mistake2} conflict with \Cref{CfTheta}, and in fact they are false. The mistake comes from the fact that the proofs use an incorrect definition of the finite derivative: the correct operation is $\mathcal{D}_f$ as introduced in \Cref{Df}, but the proofs use another operation called $D$ and defined as follows.

\begin{definition}[\cite{Huang}]\label{D}
	Over a binary alphabet $\mathpzc{A}$, we define the map
	\begin{align*}
		D~:~~~~~~~~~ \mathcal{C}_f^1 & \longrightarrow \mathpzc{A}^* \\
    	a_1^{p_1}...a_n^{p_n} & \longmapsto \textup{cut}'(p_1)~p_2~...~p_{n-1}~\textup{cut}'(p_n)
	\end{align*}
	where $\textup{cut}'(p) := \begin{cases} \varepsilon \textup{ if } 1 \leq p < b \\ b \textup{ if } p = b \end{cases}$, with the convention that $D(\varepsilon) = \varepsilon$.
\end{definition}

The map $D$ is identical to $\mathcal{D}_f$ when $a = b-1$, but not anymore when $a < b-1$. In this last case the consequence is that the set of "smooth words" defined by $D$ is bigger than $\mathcal{C}_f^\infty$ (hence $\rho' > \rho$) and contains finite words that cannot occur in smooth sequences, as illustrated in the following example.

\begin{example}
	Over $\{1,4\}$, the word $u := 4^4 1^4 4^4 1^4 4^3$ is "smooth" with respect to $D$ since ${D^3(u) = D^2(4^4) = D(4) = \varepsilon}$, but every occurrence of $u$ in a smooth sequence must be followed by a $4$ and the word $u4$ cannot occur in a smooth sequence. However, $u$ is not smooth with respect to $\mathcal{D}_f$ since $\mathcal{D}_f(u) = 4^5 \notin \mathcal{C}_f^1$.
\end{example}

\subsection{New results}

In this paper we present results regarding \Cref{mixedCf,CfTheta}.

\subsubsection{The language of smooth sequences}

In \Cref{factors} we show that the smooth words are exactly the factors of smooth sequences.

\begin{theorem}\label{Th1}
	Over any binary alphabet, we have $\mathscr{L}(\mathcal{C}^\infty) = \mathcal{C}_f^\infty$, and in particular $p_{\mathcal{C}^\infty}(n) = p_{\mathcal{C}_f^\infty}(n)$.
\end{theorem}

This does not prove \Cref{mixedCf} but it motivates even more the study of the language $\mathcal{C}_f^\infty$ and its complexity.

\subsubsection{The complexity of smooth words}

In \Cref{bispecials} we extend the description of the bispecial smooth words given in \cite{Weakley} to any binary alphabet, which allows us in \Cref{bounds} to obtain several results regarding \Cref{CfTheta}. We first prove the lower bound over any binary alphabet.

\begin{theorem}\label{Th2}
	Over a binary alphabet $\mathpzc{A} = \{a,b\}$, we have $p_{\mathcal{C}_f^\infty}(n) = \Omega(n^\rho)$ where $\rho := \frac{\log(a+b)}{\log\left(\frac{a+b}{2}\right)}$.
\end{theorem}

We also obtain several upper bounds, based on the same argument.

\begin{theorem}\label{Th3}
	Let $\mathpzc{A} = \{a,b\}$ be a binary alphabet.
	
	\textup{(i)} We have $p_{\mathcal{C}_f^\infty}(n) = \mathcal{O}(n^\zeta)$ where $\zeta := 1+\frac{\log(2)}{\log\left(\frac{a(2b+1)}{a+b+1}\right)}$.
	
	\textup{(ii)} If $\mathpzc{A}$ is even, then we have $p_{\mathcal{C}_f^\infty}(n) = \mathcal{O}(n^\rho)$ where $\rho := \frac{\log(a+b)}{\log\left(\frac{a+b}{2}\right)}$.
	
	\textup{(iii)} If $\mathpzc{A}$ is odd, then we have $p_{\mathcal{C}_f^\infty}(n) = \mathcal{O}(n^\xi)$ where $\xi := \frac{\log(2\lambda)}{\log(\lambda)}$ and
	$$\lambda := \begin{cases} \frac{1+\sqrt{2b-1}}{2} \textit{ if } a = 1, \\ \textit{the dominant root of } X^3 - \frac{a+b}{2} X^2 + \frac{(b-a)^2}{4} \textit{ otherwise.} \end{cases}$$
\end{theorem}

\Cref{Th2,Th3} (ii) prove the conjecture over even alphabets.

\begin{corollary}
	\Cref{CfTheta} holds over even alphabets.
\end{corollary}

\begin{remark}\label{rankingUpperBounds}
	The upper bound of \Cref{CfTheta} remains unproven over odd and mixed alphabets. The proven upper bounds are now
	\begin{itemize}
		\item $\mathcal{O}(n^\beta)$ from \Cref{firstPoly} \textup{(ii)}. In particular, we have $\frac{\beta}{\rho} > 1.79$ for all binary alphabet and $\displaystyle\lim_{b-a \rightarrow \infty} \frac{\beta}{\rho} = +\infty$.
		
		\item $\mathcal{O}(n^\zeta)$ from \Cref{Th3} \textup{(i)}. In particular, we have $\frac{\zeta}{\rho} < 2$ for all binary alphabet and $\displaystyle\lim_{b \rightarrow \infty} \frac{\zeta}{\rho} = 1+\frac{\log(2)}{\log(2a)}$.
		
		\item $\mathcal{O}(n^\xi)$ from \Cref{Th3} \textup{(iii)} over odd alphabets. In particular, we have $\frac{\xi}{\rho} < 1.23$ for all odd alphabet and $\displaystyle\lim_{a+b \rightarrow \infty} \frac{\xi}{\rho} = 1$.
	\end{itemize}
	In \Cref{unsufficient} we explain that $\xi$ is the best bound we can obtain with our method over odd alphabets. Furthemore, technical arguments yield $\rho < \xi < \zeta < \beta$. To illustrate this comparison, we display (the rounded values of) these exponents over several odd alphabets in \Cref{oddExponents} and over several mixed alphabets in \Cref{mixedExponents}.
	\begin{figure}[H]
		\begin{center}
		\begin{tabular}{|l|c|c|c|c|c|c|c|c|c|c|c|c|c|}
			\hline
  			$\mathpzc{A}$ & $\{1,3\}$ & $\{1,5\}$ & $\{3,5\}$ & $\{1,7\}$ & $\{3,7\}$ & $\{5,7\}$ & $\{1,9\}$ & $\{3,9\}$ & $\{5,9\}$ & $\{7,9\}$ & $\{1,11\}$ & $\{3,11\}$ \\
			\hline
			\hline
			$\rho$ & 2 & 1.631 & 1.5 & 1.5 & 1.431 & 1.387 & 1.431 & 1.387 & 1.356 & 1.333 & 1.387 & 1.356 \\
			\hline
			$\xi$ & 2.44 & 2 & 1.506 & 1.831 & 1.44 & 1.388 & 1.737 & 1.397 & 1.358 & 1.334 & 1.675 & 1.366 \\
			\hline
			$\zeta$ & 3.06 & 2.534 & 1.533 & 2.357 & 1.492 & 1.396 & 2.268 & 1.469 & 1.376 & 1.337 & 2.215 & 1.454 \\
			\hline
			$\beta$ & 7.129 & 7.658 & 2.96 & 8.193 & 3.195 & 2.6 & 8.656 & 3.383 & 2.734 & 2.465 & 9.056 & 3.54 \\
			\hline
		\end{tabular}
		\caption{Comparison of upper bounds over several odd alphabets.}
		\label{oddExponents}
		\end{center}
		\begin{center}
		\begin{tabular}{|l|c|c|c|c|c|c|c|c|c|c|c|c|}
			\hline
  			$\mathpzc{A}$ & $\{1,2\}$ & $\{2,3\}$ & $\{1,4\}$ & $\{3,4\}$ & $\{2,5\}$ & $\{4,5\}$ & $\{1,6\}$ & $\{3,6\}$ & $\{5,6\}$ & $\{2,7\}$ & $\{4,7\}$ & $\{6,7\}$ \\
			\hline
			\hline
			$\rho$ & 2.71 & 1.756 & 1.756 & 1.553 & 1.553 & 1.461 & 1.553 & 1.461 & 1.407 & 1.461 & 1.407 & 1.37 \\
			\hline
			$\zeta$ & 4.106 & 1.818 & 2.71 & 1.57 & 1.685 & 1.468 & 2.428 & 1.509 & 1.41 & 1.631 & 1.431 & 1.373 \\
			\hline
			$\beta$ & 7.228 & 3.302 & 7.374 & 2.813 & 3.726 & 2.623 & 7.934 & 3.085 & 2.521 & 4.04 & 2.817 & 2.457 \\
			\hline
		\end{tabular}
		\caption{Comparison of upper bounds over several mixed alphabets.}
		\label{mixedExponents}
		\end{center}
	\end{figure}
\end{remark}

With assumptions on the letter frequencies in $\mathcal{C}_f^\infty$, we obtain other upper bounds.

\begin{theorem}\label{Th4}
	Let $\mathpzc{A} = \{a,b\}$ be a binary alphabet and let $\rho := \frac{\log(a+b)}{\log\left(\frac{a+b}{2}\right)}$.
	
	\textup{(i)} If $\displaystyle \liminf_{u \in \mathcal{C}_f^\infty} \frac{\lvert u \rvert_b}{\lvert u \rvert} > \frac{1}{2}-\phi$ for some $\phi \in \left(0,\frac{1}{2}\right)$, then we have $p_{\mathcal{C}_f^\infty}(n) = \mathcal{O}(n^\vartheta)$ where $\vartheta := 1+\frac{\log(2)}{\log\left(\frac{a+b}{2}-(b-a)\phi\right)}$.
	
	\textup{(ii)} If $\displaystyle \lim_{u \in \mathcal{C}_f^\infty} \frac{\lvert u \rvert_b}{\lvert u \rvert} = \frac{1}{2}$, then we have $p_{\mathcal{C}_f^\infty}(n) = \mathcal{O}(n^{\rho+\phi})$ for all $\phi > 0$.
	
	\textup{(iii)} If there exists $\phi > 0$ such that $\left\lvert \frac{\lvert u \rvert_b}{\lvert u \rvert} - \frac{1}{2} \right\rvert = \mathcal{O}\left(\log(\lvert u \rvert)^{-(1+\phi)}\right)$ for all $u \in \mathcal{C}_f^\infty$, then we have $p_{\mathcal{C}_f^\infty}(n) = \mathcal{O}(n^\rho)$.
\end{theorem}

\begin{remark}
	\Cref{Th4} \textup{(i)} is a slight improvement of \Cref{freqCf} \textup{(ii)} and a rectification of \Cref{mistake1} \textup{(ii)}.
\end{remark}

\begin{remark}
	\Cref{Th4} \textup{(ii)} is a corollary of \Cref{Th4} \textup{(i)}. The assumption $\displaystyle \lim_{u \in \mathcal{C}_f^\infty} \frac{\lvert u \rvert_b}{\lvert u \rvert} = \frac{1}{2}$ is easy to prove over even alphabets, false over odd alphabets (for example the smooth sequence $\kappa_{3,1}$ described in the introduction has frequencies of $1$ and $3$ that are not equal) and conjectured over mixed alphabets.
\end{remark}

\subsubsection{Consequences over $\mathpzc{A} = \{1,2\}$}

In \Cref{consequences} we focus on the alphabet $\{1,2\}$ to study the Oldenburger-Kolakoski sequence $\kappa$. We first combine our previous results with machine computation from \cite{Rao,Nilsson} (see \Cref{freq12}) to obtain the following bounds of $p_{\mathcal{C}_f^\infty}(n)$.

\begin{corollary}\label{freqMixed12}
	Over $\{1,2\}$, we have $p_{\mathcal{C}_f^\infty}(n) = \Omega\left(n^\rho\right)$ and $p_{\mathcal{C}_f^\infty}(n) = \mathcal{O}\left(n^{\rho + 0.00023}\right)$ where $\rho := \frac{\log(3)}{\log(3/2)} \approx 2.70951$.
\end{corollary}

\begin{remark}
	Carpi conjectured in \cite{Carpi} that $\left\lvert \lvert u \rvert_2 - \lvert u \rvert_1 \right\rvert = \mathcal{O}\left(\lvert u \rvert^{1-\phi}\right)$ for some $\phi > 0$, which is equivalent to having $\left\lvert \frac{\lvert u \rvert_2}{\lvert u \rvert} - \frac{1}{2} \right\rvert = \mathcal{O}\left(\lvert u \rvert^{-\phi}\right)$. It is stronger that the assumption of \Cref{Th4} \textup{(iii)} so it would imply that $p_{\mathcal{C}_f^\infty}(n) = \mathcal{O}\left(n^\rho\right)$.
\end{remark}

Then we show a dichotomy that links \Cref{mixedCf} to the morphicity of the Oldenburger-Kolakoski sequence $\kappa$.

\begin{theorem}\label{dichotomy}
	Let $\mathpzc{A} = \{1,2\}$.
	
	\textup{(i)} If $\mathscr{L}(\kappa) = \mathcal{C}_f^\infty$, then $\kappa$ is recurrent but not linearly recurrent and not morphic, and we have $p_{\kappa}(n) = \Omega\left(n^\rho\right)$ and $p_{\kappa}(n) = \mathcal{O}\left(n^{\rho+0.00023}\right)$ where ${\rho := \frac{\log(3)}{\log(3/2)} \approx 2.70951}$.
	
	\textup{(ii)} If $\mathscr{L}(\kappa) \subsetneq \mathcal{C}_f^\infty$ and $\kappa$ is uniformly recurrent, then $\kappa$ is linearly recurrent and morphic, and we have $p_{\kappa}(n) = \Theta(n)$.
\end{theorem}

This strengthens the two conjectures that $\mathscr{L}(\kappa) = \mathcal{C}_f^\infty$ and that $\kappa$ is not morphic.

\section{The language of smooth sequences}\label{factors}

In this section we prove \Cref{Th1} using a right version of smooth words called \textit{r-smooth words}, originally introduced in \cite{BMP}, which will play the role of prefixes in smooth sequences.

\subsection{r-smooth words}

\begin{definition}
	Over a binary alphabet $\mathpzc{A}$, we say that a word $u \in \mathpzc{A}^*$ is \textup{right-derivable} if it is factorized as $a_1^{p_ 1} ... a_n^{p_ n}$ where $p_i \in \mathpzc{A}$ for all $i \in \llbracket 1,n-1 \rrbracket$ and $p_n \in \llbracket 1,b \rrbracket$. We write $\mathcal{C}^1_r$ the set of right-derivable words and we define the \textup{right-derivative} as the map
	\begin{align*}
		\mathcal{D}_r~:~~~~~~~~~~ \mathcal{C}^1_r & \longrightarrow \mathpzc{A}^* \\
    	a_1^{p_1}...a_n^{p_n} & \longmapsto p_1~...~p_{n-1}~\textup{cut}(p_n)
	\end{align*}
	where the function $\textup{cut}$ is the same as in $\mathcal{D}_f$ and $\mathcal{D}_r(a_1^{p_1}) = \textup{cut}(p_1)$.
\end{definition}

\begin{example}
	Over $\{1,2\}$, we have $\mathcal{D}_r(\varepsilon) = \mathcal{D}_r(2) = \varepsilon$, $\mathcal{D}_r(21) = 1$ and $\mathcal{D}_r(211) = 12$.
\end{example}

\begin{definition}
	For $n \geq 2$, we define by induction the set of words that are right-derivable $n$ times
	\begin{equation*}
		\mathcal{C}_r^n := \left\{u \in \mathcal{C}^1_r ~\middle |~ \mathcal{D}_r(u) \in \mathcal{C}_r^{n-1} \right\}.
	\end{equation*}
	We then define the set of \textup{r-smooth words}, that is the set of words that are infinitely right-derivable
	\begin{equation*}
		\mathcal{C}_r^\infty := \displaystyle \bigcap_{n \geq 1} \mathcal{C}_r^n.
	\end{equation*}
\end{definition}

Similarly to $\mathcal{C}_f^\infty$, one can easily prove by induction that $u \sqsubset_p v \in \mathcal{C}_r^\infty$ implies $u \in \mathcal{C}_r^\infty$, and that the language $\mathcal{C}_r^\infty$ is right-extendable.

\subsection{Prefixes and factors of smooth sequences}

Let us prove \Cref{Th1} by showing the intermediate fact that $\mathcal{C}_r^\infty$ is exactly the set of prefixes of smooth sequences.

\begin{proof}[Proof of \Cref*{Th1}]
	Let $\mathpzc{A} = \{a,b\}$ be a binary alphabet. We already observed that $\mathscr L(\mathcal{C}^\infty) \subseteq \mathcal{C}_f^\infty$, so it remains to prove that $\mathcal{C}_f^\infty \subseteq \mathscr{L}(\mathcal{C}^\infty)$.
	
	The first step is to show that every smooth word is the factor of an r-smooth word. We are going to show that, for all $u \in \mathcal{C}_f^\infty$, there exists $\ell \geq a+b$ and $v \in \mathpzc{A}^\ell$ such that $vu \in \mathcal{C}_r^\infty$. We proceed by induction on the height of $u$: first, if $u$ has height $0$, i.e., $u = \varepsilon$, then the word $v = a^b b^a$ suffices. Now let $h \geq 0$ be such that, for all $u \in \mathcal{C}_f^\infty$ of height $h$, there exists $\ell \geq a+b$ and $v \in \mathpzc{A}^\ell$ such that $vu \in \mathcal{C}_r^\infty$. Let $u \in \mathcal{C}_f^\infty$ be of height $h+1$ and consider its canonical factorization $u = a_1^{p_1} ... a_n^{p_n}$ with $n \geq 1$. The word $\mathcal{D}_f(u)$ is smooth and has height $h$ so, by hypothesis, there exists $\ell \geq a+b$ and $v \in \mathpzc{A}^\ell$ such that $v \mathcal{D}_f(u) \in \mathcal{C}_r^\infty$. In particular, $v \in \mathcal{C}_r^\infty$ and it is long enough to contain at least $a$ occurrences of the letter $b$. We deduce that
	$$\displaystyle\sum_{i=1}^\ell v_i = a\left\lvert v \right\rvert_a + b\lvert v \rvert_b = a\left(\ell - \lvert v \rvert_b\right) + b\lvert v \rvert_b = a\ell + (b-a)\lvert v \rvert_b \geq a(a+b) + (b-a)a \geq (a+b)+a.$$
	
	If $p_1 \leq a$, we have $\mathcal{D}_f(u) = p_2 ... p_{n-1} \textup{cut}(p_n)$ and we set $w := c_1^{v_1} ... c_\ell^{v_\ell} c_{\ell+1}^{p_2} ... c_{\ell+n-1}^{p_n}$ where $c_{i+1} = \overline{c_i}$ and ${c_{\ell+1} = a_2}$. We first observe that $w = v' u$ where ${v' := c_1^{v_1} ... c_{\ell-1}^{v_{\ell-1}} c_\ell^{v_\ell-p_1}}$ and that $\left\lvert v' \right\rvert = \left(\displaystyle\sum_{i=1}^\ell v_i\right) - p_1 \geq (a+b+a)-a = a+b$. We also observe that $w \in \mathcal{C}^1_r$ and that $\mathcal{D}_r(w) = v \mathcal{D}_f(u)$, which implies that $w = v'u \in \mathcal{C}_r^\infty$.
	
	If $p_1 > a$, we have $\mathcal{D}_f(u) = b p_2 ... p_{n-1} \textup{cut}(p_n)$ and we set $w := c_1^{v_1} ... c_\ell^{v_\ell} c_{\ell+1}^b c_{\ell+2}^{p_2} ... c_{\ell+n}^{p_n}$ where $c_{i+1} = \overline{c_i}$ and $c_{\ell+1} = a_1$. We first observe that $w = v' u$ where ${v' := c_1^{v_1} ... c_\ell^{v_\ell} c_{\ell+1}^{b-p_1}}$ and that $\left\lvert v' \right\rvert \geq \displaystyle\sum_{i=1}^\ell v_i \geq a+b+a$. We also observe that $w \in \mathcal{C}^1_r$ and that $\mathcal{D}_r(w) = v \mathcal{D}_f(u)$, which implies that $w = v'u \in \mathcal{C}_r^\infty$.
	
	The second step is to show that a sequence $x \in \mathpzc{A}^\mathbb{N}$ such that all its prefixes are r-smooth is smooth. We show by induction on $n$ that, for all $n \geq 1$, a sequence $x \in \mathpzc{A}^\mathbb{N}$ such that all its prefixes are r-smooth satisfies $x \in \mathcal{C}^n$. If $n = 1$, we easily notice that a sequence such that all its prefixes are r-smooth is derivable. Now let $n \geq 1$ be such that a sequence $x \in \mathpzc{A}^\mathbb{N}$ such that all its prefixes are r-smooth satisfies $x \in \mathcal{C}^n$, and let $x \in \mathpzc{A}^\mathbb{N}$ be such that all its prefixes are r-smooth. In particular, the words $\mathcal{D}_r(x_{[0,bm)})$ are r-smooth for all $m \geq 0$ and are prefixes of arbitrarily large length of $\mathcal{D}(x)$. We deduce that all prefixes of $\mathcal{D}(x)$ are r-smooth, and by hypothesis we have $\mathcal{D}(x) \in \mathcal{C}^n$ so $x \in \mathcal{C}^{n+1}$.
	
	The third step is to show that every r-smooth word is the prefix of a smooth sequence. If $v$ is r-smooth, the fact that the language $\mathcal{C}_r^\infty$ is right-extendable provides a sequence of r-smooth words of strictly increasing length, starting with $v$, which are all prefixes of a sequence $x \in \mathpzc{A}^\mathbb{N}$. By construction we have $v \sqsubset_p x$, and all prefixes of $x$ are r-smooth so $x$ is smooth.
	
	Finally, if $u \in \mathcal{C}_f^\infty$, we proved that there exists $v \in \mathcal{C}_r^\infty$ such that $u \sqsubset v$ and that there exists $x \in \mathcal{C}^\infty$ such that $v \sqsubset_p x$, so we deduce that $u \sqsubset x$. Therefore, we have $\mathcal{C}_f^\infty \subseteq \mathscr{L}(\mathcal{C}^\infty)$.
\end{proof}

\section{Bispecial smooth words}\label{bispecials}

In this section we detail how to compute $p_{\mathcal{C}_f^\infty}(n)$ from the bispecial smooth words, extending the work of Weakley \cite{Weakley} to any binary alphabet. First we split the strong and weak bispecial smooth words into five families, each forming an infinite binary tree introduced in \Cref{families}. Then we show in \Cref{p3p} that it suffices to study a single family of strong bispecial smooth words to bound $p_{\mathcal{C}_f^\infty}(n)$.

\subsection{Finite primitives}

In order to determine the bispecial smooth words, we introduce the notion of \textit{finite primitives}. This starts with the observation that the two longest words with derivative $u \in \mathpzc{A}^n$ are canonically factorized as ${a_0^a a_1^{u_1} a_2^{u_2} ... a_n^{u_n} a_{n+1}^a}$.

\begin{definition}
	Over a binary alphabet $\mathpzc{A} = \{a,b\}$, the \textup{finite primitives} are the maps
	\begin{align*}
		\mathcal{P}_{f,a}~:~~~~~~~~\mathpzc{A}^* & \longrightarrow \mathpzc{A}^+ & \mathcal{P}_{f,b}~: \mathpzc{A}^* & \longrightarrow \mathpzc{A}^+ \\
    	u_1...u_n & \longmapsto \begin{cases} a^ab^{u_1}a^{u_2}...a^{u_n} b^a \textup{ if } n \textup{ is even,} \\
   		a^ab^{u_1}a^{u_2}...b^{u_n} a^a \textup{ otherwise.} \end{cases} & u & \longmapsto \overline{\mathcal{P}_{f,a}(u)}
	\end{align*}
\end{definition}

\begin{example}
	Over $\{1,2\}$, we have $\mathcal{P}_{f,1}(2) = 1221$ and $\mathcal{P}_{f,2}(2) = 2112$. Over $\{2,4\}$, we have ${\mathcal{P}_{f,2}(42) = 2244442244}$.
\end{example}

We make three important remarks.

\begin{remark}\label{prefixP}
	Let $u \in \mathpzc{A}^*$ and let $c \in \mathpzc{A}$.
	
	\textup{(i)} We have $\left\lvert \mathcal{P}_{f,c}(u) \right\rvert = a \lvert u \rvert_a + b \lvert u \rvert_b + 2a$.
	
	\textup{(ii)} We have $\mathcal{D}_f\left(\mathcal{P}_{f,c}(u)\right) = u$. We deduce that $u$ is smooth if and only if $\mathcal{P}_{f,c}(u)$ is smooth.
	
	\textup{(iii)} If $u \sqsubset_p v \in \mathpzc{A}^*$, then $\mathcal{P}_{f,c}(u) \sqsubset_p \mathcal{P}_{f,c}(v)$.
\end{remark}

\subsection{Description of the bispecial smooth words}

In this subsection we show that each bispecial smooth word can be expressed with finite primitives and reduced to a "short" bispecial word which has the same type.

\begin{definition}
	Over a binary alphabet $\mathpzc{A} = \{a,b\}$, we say that a word of the form $c^n$
where $c \in \mathpzc{A}$ and $n \in \llbracket 0,b-1 \rrbracket$ is \textup{short}, and any word that is not short is \textup{long}.
\end{definition}

\begin{lemma}\label{Dbispecial}
	Let $\mathpzc{A} = \{a,b\}$ be a binary alphabet, let $u \in \mathpzc{A}^*$ and let $c \in \mathpzc{A}$.
	
	\textup{(i)}  $u$ is a bispecial smooth word if and only if $\mathcal{P}_{f,c}(u)$ is a bispecial smooth word, and if they are bispecial then they have the same type.
	
	\textup{(ii)} If $u$ is a long bispecial smooth word, then $u = \mathcal{P}_{f,u_1}\left(\mathcal{D}_f(u)\right)$ and $\mathcal{D}_f(u)$ is also a bispecial smooth word with the same type as $u$.
\end{lemma}
\begin{proof}
	(i) Set $n := \lvert u \rvert$ and $v := \mathcal{P}_{f,c}(u)$, so that the canonical factorization of $v$ is ${a_0}^a {a_1}^{u_1} ... {a_n}^{u_n} {a_{n+1}}^a$. Then the proof relies on the four following equivalences:
	\begin{align}
		aua \in \mathcal{C}_f^\infty \iff \overline{a_0} v \overline{a_{n+1}} \in \mathcal{C}_f^\infty \label{eq1} \\
		aub \in \mathcal{C}_f^\infty \iff \overline{a_0} v a_{n+1} \in \mathcal{C}_f^\infty  \label{eq2} \\
		bua \in \mathcal{C}_f^\infty \iff a_0 v \overline{a_{n+1}} \in \mathcal{C}_f^\infty \label{eq3} \\
		bub \in \mathcal{C}_f^\infty \iff a_0 v a_{n+1} \in \mathcal{C}_f^\infty \label{eq4}
	\end{align}
	These equivalences come from the simple equivalence: if $w \in \mathpzc{A}^*$, then $w \in \mathcal{C}_f^\infty$ if and only if $w \in \mathcal{C}_f^1$ and $\mathcal{D}_f(w) \in \mathcal{C}_f^\infty$. Now we observe that $\mathcal{D}_f(\overline{a_0} v \overline{a_{n+1}}) = aua$, which yields (\ref{eq1}); that $\mathcal{D}_f(\overline{a_0} v a_{n+1}) = aub$ yields (\ref{eq2}), ${\mathcal{D}_f(a_0 v \overline{a_{n+1}}) = bua}$ yields (\ref{eq3}) and ${\mathcal{D}_f(a_0 v a_{n+1}) = bub}$ yields (\ref{eq4}).
	
	Now suppose that $u$ is a bispecial smooth word. We have $au \in \mathcal{C}_f^\infty$ and the language $\mathcal{C}_f^\infty$ is right-extendable so we have $aua \in \mathcal{C}_f^\infty$ or $aub \in \mathcal{C}_f^\infty$. Then (\ref{eq1}) and (\ref{eq2}) imply that $\overline{a_0} v \overline{a_{n+1}} \in \mathcal{C}_f^\infty$ or $\overline{a_0} v a_{n+1} \in \mathcal{C}_f^\infty$, and we deduce that $\overline{a_0} v \in \mathcal{C}_f^\infty$. In the exact same way, we obtain $a_0 v \in \mathcal{C}_f^\infty$ from $bu \in \mathcal{C}_f^\infty$, $v\overline{a_{n+1}} \in \mathcal{C}_f^\infty$ from $ua \in \mathcal{C}_f^\infty$ and $v a_{n+1} \in \mathcal{C}_f^\infty$ from $ub \in \mathcal{C}_f^\infty$, therefore $v$ is a bispecial smooth word. Conversely, if $v$ is a bispecial smooth word, the same argument with the four equivalences (\ref{eq1}), (\ref{eq2}), (\ref{eq3}) and (\ref{eq4}) yields that $u$ is a bispecial smooth word. Finally, if $u$ and $v$ are bispecial smooth words, the four equivalences (\ref{eq1}), (\ref{eq2}), (\ref{eq3}) and (\ref{eq4}) directly imply that $u$ and $v$ have the same type.
	
	(ii) Let $a_1^{p_1} ... a_n^{p_n}$ be the canonical factorization of $u$, where $n \geq 1$. If $n = 1$, $u$ is long so we necessarily have $u = c^b$ where $c \in \mathpzc{A}$, but the smooth word $c^b$ is not bispecial since $c^b c \notin \mathcal{C}_f^\infty$. Therefore we have $n \geq 2$. Now, we have $a_1^{p_1+1} ... a_n^{p_n} = a_1 u \in \mathcal{C}_f^\infty$ so $p_1 + 1 \leq b$, and $\overline{a_1} a_1^{p_1} ... a_n^{p_n} = \overline{a_1} u \in \mathcal{C}_f^\infty$ so $p_1 \in \{a,b\}$, therefore $p_1 = a$; and symmetrically we obtain $p_n = a$. We deduce that $u = a_1^a a_2^{p_2} ... a_{n-1}^{p_{n-1}} a_n^a$ and $\mathcal{D}_f(u) = p_2 ... p_{n-1}$, which yields $u = \mathcal{P}_{f,u_1}\left(\mathcal{D}_f(u)\right)$ where $a_1 = u_1$. Then (i) states that $\mathcal{D}_f(u)$ is a bispecial smooth word with the same type as $u$.
\end{proof}

The next definition naturally follows.

\begin{definition}
	Let $\mathpzc{A}$ be a binary alphabet and let $u$ be a bispecial smooth word. Then the \textup{degree} of $u$ is the first integer $i \geq 0$ such that $\left(\mathcal{D}_f\right)^i(u)$ is a short bispecial smooth word, and we call $\left(\mathcal{D}_f\right)^i(u)$ the \textup{root} of $u$.
\end{definition}

With these definitions, \Cref{Dbispecial} directly yields a characterization of the bispecial factors.

\begin{corollary}\label{bispecialsmooth}
	Let $\mathpzc{A} = \{a,b\}$ be a binary alphabet. Then the bispecial smooth words of degree $i$ are the words of the form $\mathcal{P}_{f,t_0} \mathcal{P}_{f,t_1} ... \mathcal{P}_{f,t_{i-1}} (u)$ where $u$ is a short bispecial smooth word and $t_j \in \mathpzc{A}$ for all $j \in \llbracket 0,i-1\rrbracket$. Moreover, each bispecial smooth word has the same type as its root $u$.
\end{corollary}

It remains to identify the short bispecial smooth words and their type.

\begin{lemma}
	Let $\mathpzc{A} = \{a,b\}$ be a binary alphabet.
	
	\textup{(i)} If $a = b-1$, then $\varepsilon$ is the unique short strong bispecial smooth word and there is no short weak bispecial smooth word.
	
	\textup{(ii)} If $a < b-1$, then the short strong bispecial smooth words are $\varepsilon$, $a^a$ and $b^a$; and the short weak bispecial smooth words are $a^{b-1}$ and $b^{b-1}$.
\end{lemma}
\begin{proof}
	Let $c^n$ be a short word, where $c \in \mathpzc{A}$ and $n \in \llbracket 0,b-1\rrbracket$. It is smooth so we check its bi-extensions in $\mathcal{C}_f^\infty$.
	\begin{center}
		\begin{tabular}{|l|c|c|c|c|c|}
  		\hline
  		& $c c^n c$ & $c c^n \overline{c}$ & $\overline{c} c^n c$ & $\overline{c} c^n \overline{c}$ & type of $c^n$ \\
  		\hline
  		\hline
  		$n = 0$ or $n = a < b-1$ & $\in \mathcal{C}_f^\infty$ & $\in \mathcal{C}_f^\infty$ & $\in \mathcal{C}_f^\infty$ & $\in \mathcal{C}_f^\infty$ & strong \\
  		\hline
  		$n \in \llbracket 1,a-1 \rrbracket \cup \llbracket a+1,b-2 \rrbracket$ & $\in \mathcal{C}_f^\infty$ & $\in \mathcal{C}_f^\infty$ & $\in \mathcal{C}_f^\infty$ & $\notin \mathcal{C}_f^\infty$ & neutral \\
  		\hline
  		$n = a = b-1$ & $\notin \mathcal{C}_f^\infty$ & $\in \mathcal{C}_f^\infty$ & $\in \mathcal{C}_f^\infty$ & $\in \mathcal{C}_f^\infty$ & neutral \\
  		\hline
  		$n = b-1 > a$ & $\notin \mathcal{C}_f^\infty$ & $\in \mathcal{C}_f^\infty$ & $\in \mathcal{C}_f^\infty$ & $\notin \mathcal{C}_f^\infty$ & weak \\
  		\hline
		\end{tabular}
	\end{center}
\end{proof}

\begin{definition}\label{families}
	Over a binary alphabet $\mathpzc{A} = \{a,b\}$, we define the infinite binary tree $T$ as follows: its root is $\varepsilon$ and the two children of a vertex $u \in \mathpzc{A}^*$ are $\mathcal{P}_{f,a}(u)$ and $\mathcal{P}_{f,b}(u)$. By \Cref{bispecialsmooth}, the vertices of $T$ are all the strong bispecial smooth words of root $\varepsilon$. If $a < b-1$, we define in the same way the following infinite binary trees:
	\begin{itemize}
		\item $T^{(1)}$ is the tree of strong bispecial smooth words of root $a^a$,
	
		\item $T^{(2)}$ is the tree of strong bispecial smooth words of root $b^a$,
	
		\item $T^{(3)}$ is the tree of weak bispecial smooth words of root $a^{b-1}$,
	
		\item $T^{(4)}$ is the tree of weak bispecial smooth words of root $b^{b-1}$.
	\end{itemize}
\end{definition}

\begin{example}
	Over alphabets of the form $\mathpzc{A} = \{a,a+1\}$, the strong bispecial smooth words are gathered in the tree $T$. We display it over $\{1,2\}$ in \Cref{G012}.
	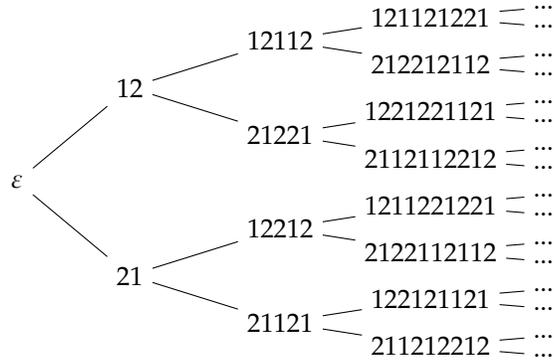
\begin{figure}[H]
		\begin{center}
		\begin{tikzpicture}
 		\node at (4,0) {$\varepsilon$} [grow = 0, sibling distance = 2.5cm]
 			child {node {21} [sibling distance = 1.25cm, level distance = 2cm]
 				child {node{21121} [sibling distance = 0.625cm, level distance = 2cm]
 					child{node{211212212} [sibling distance = 0.25cm, level distance = 1.5cm]
 						child{node{...}}
 						child{node{...}}}
 					child{node{122121121} [sibling distance = 0.25cm, level distance = 1.5cm]
 						child{node{...}}
 						child{node{...}}}}
    			child {node{12212} [sibling distance = 0.625cm, level distance = 2cm]
    				child {node{2122112112} [sibling distance = 0.25cm, level distance = 1.5cm]
    					child{node{...}}
 						child{node{...}}}
    				child {node{1211221221} [sibling distance = 0.25cm, level distance = 1.5cm]
    					child{node{...}}
 						child{node{...}}}}}
    		child {node {12} [sibling distance = 1.25cm, level distance = 2cm]
    			child {node{21221} [sibling distance = 0.625cm, level distance = 2cm]
    				child {node{2112112212} [sibling distance = 0.25cm, level distance = 1.5cm]
    					child{node{...}}
 						child{node{...}}}
    				child {node{1221221121} [sibling distance = 0.25cm, level distance = 1.5cm]
    					child{node{...}}
 						child{node{...}}}}
    			child {node{12112} [sibling distance = 0.625cm, level distance = 2cm]
    				child{node{212212112} [sibling distance = 0.25cm, level distance = 1.5cm]
    					child{node{...}}
 						child{node{...}}}
    				child{node{121121221} [sibling distance = 0.25cm, level distance = 1.5cm]
    					child{node{...}}
 						child{node{...}}}}};
		\end{tikzpicture}
		\end{center}
		\caption{The tree $T$ over $\{1,2\}$.}
		\label{G012}
	\end{figure}
\end{example}

\subsection{Computation of the complexity}

Computing $p_{\mathcal{C}_f^\infty}(n)$ with \Cref{pfrombsw} involves the five trees $T$, $T^{(1)}$, $T^{(2)}$, $T^{(3)}$ and $T^{(4)}$ (unless $b = a+1$), but in this subsection we show that it suffices to study $T$ in order to bound $p_{\mathcal{C}_f^\infty}(n)$.

\begin{definition}\label{Tcontribution}
	For all $i \geq 0$, we define $T_i$ as the $i$-th generation of $T$, i.e., the set of bispecial smooth words of root $\varepsilon$ and degree $i$. In order to estimate the contribution of $T$ to $p_{\mathcal{C}_f^\infty}(n)$, we define the following numbers for all $n,i \geq 0$:
	\begin{align*}
		& b_i(n) := \left\lvert T_i \cap \mathpzc{A}^n \right\rvert, & & s_i(n) := \displaystyle \sum_{m=0}^{n-1} b_i(m), \\
		& p_i(n) := \displaystyle \sum_{m=0}^{n-1} s_i(m), & & p(n) := \displaystyle \sum_{i=0}^\infty p_i(n).
	\end{align*}
	Note that $p(n)$ is well defined because $b_i(n)$ is eventually zero when $i$ grows. Also, summation provides
	$${p_i(n) = \displaystyle\sum_{u \in T_i} \max(0,n-\lvert u \rvert - 1)}$$
	for all $n,i \geq 0$.
\end{definition}

\begin{proposition}\label{p3p}
	Over a binary alphabet $\mathpzc{A} = \{a,b\}$, we have
	\begin{equation*}
		1+n+p(n) \leq p_{\mathcal{C}_f^\infty}(n) \leq 1+n+3p(n)
	\end{equation*}
	for all $n \geq 0$ .
\end{proposition}
\begin{proof}
	Let $s(n)$ be the first finite difference and let $b(n)$ be the second finite difference of $p_{\mathcal{C}_f^\infty}(n)$, as defined in \Cref{complexity}. If $a = b-1$, there is no weak bispecial smooth word and every strong bispecial smooth word belongs to $T$ so, with \Cref{pfrombsw} and by summation, for all $n \geq 0$ we have
	\begin{align*}
		& b(n) = \left\lvert T \cap \mathpzc{A}^n \right\rvert = \sum_{i = 0}^\infty b_i(n), \\
		& s(n) = s(0) + \displaystyle \sum_{m=0}^{n-1} b(m) = 1+\sum_{i = 0}^\infty \sum_{m=0}^{n-1} b_i(m) = 1+\sum_{i = 0}^\infty s_i(n), \\
		& p_{\mathcal{C}_f^\infty}(n) = 1 + \displaystyle \sum_{m=0}^{n-1} s(m) = 1+n+\sum_{i = 0}^\infty \sum_{m=0}^{n-1} s_i(m) = 1+n+\sum_{i = 0}^\infty p_i(n) = 1+n+p(n).
	\end{align*}
	
	If $a < b-1$, for $t \in \{1,2,3,4\}$ we define $T^{(t)}_i$, $b^{(t)}_i(n)$, $s^{(t)}_i(n)$, $p^{(t)}_i(n)$ and $p^{(t)}(n)$ exactly as for $T$. Now the strong (resp. weak) bispecial smooth words belong to $T$, $T^{(1)}$ or $T^{(2)}$ (resp. to $T^{(3)}$ or $T^{(4)}$) so \Cref{pfrombsw} yields
	\begin{equation}\label{split5}
		p_{\mathcal{C}_f^\infty}(n) = 1 + n + p(n) + p^{(1)}(n) + p^{(2)}(n) - p^{(3)}(n) - p^{(4)}(n).
	\end{equation}
	
	Now we define the bijection
	\begin{align*}
		f : ~~~~~~~~~~~~~~~~~~~~~~~~~~~~~~~ T & \longrightarrow T^{(1)} \\
		\mathcal{P}_{f,t_0} \mathcal{P}_{f,t_1} ... \mathcal{P}_{f,t_{i-1}} (\varepsilon) & \longmapsto \mathcal{P}_{f,t_0} \mathcal{P}_{f,t_1} ... \mathcal{P}_{f,t_{i-1}} (a^a)
	\end{align*}
	so that $f(T_i) = T^{(1)}_i$ for all $i \geq 0$. A quick induction on $i$ with $\varepsilon \sqsubset_p a^a$ and \Cref{prefixP} (iii) yields $u \sqsubset_p f(u)$ for all $u \in T$. We deduce that
	$$p^{(1)}_i(n) = \displaystyle\sum_{u \in T^{(1)}_i} \max\left(0,n-\lvert u \rvert - 1\right) = \sum_{u \in T_i} \max\left(0,n - \left\lvert f(u) \right\rvert - 1\right) \leq \sum_{u \in T_i} \max(0,n-\lvert u \rvert - 1) = p_i(n),$$
	for all $i \geq 0$ and all $n \geq 0$, which yields $p^{(1)}(n) \leq p(n)$ for all $n \geq 0$.

The same method yields $p^{(2)}(n) \leq p(n)$ (since $\varepsilon \sqsubset_p b^a$), ${p^{(3)}(n) \leq p^{(1)}(n)}$ (since ${a^a \sqsubset_p a^{b-1}}$) and ${p^{(4)}(n) \leq p^{(2)}(n)}$ (since ${b^a \sqsubset_p b^{b-1}}$). Finally, combining these inequalities with \Cref{split5} yields the result.
\end{proof}

\section{Bounds of the complexity}\label{bounds}

In this section we prove \Cref{Th2,Th3,Th4}. Thanks to \Cref{p3p} it remains to study the asymptotics of $p(n)$ in the tree $T$, we will do it using the following definition.

\begin{definition}\label{liLi}
	If $\mathpzc{A} = \{a,b\}$ is a binary alphabet, define $\ell_i := \displaystyle\min_{u \in T_i} \lvert u \rvert$ and $L_i := \displaystyle\max_{u \in T_i} \lvert u \rvert$.
\end{definition}

Recall that $\ell_i$ and $L_i$ already appeared in \Cref{almost12,Lili1}.

\subsection{The lower bound}

We prove here the conjectured lower bound over any binary alphabet (\Cref{Th2}) by considering the average length in each $T_i$.

\begin{proof}[Proof of \Cref*{Th2}]
	The first step is to compute the average length in generation of $T$. For $i \geq 0$, set ${f(i) := \displaystyle \sum_{m = \ell_i}^{L_i} m b_i(m) = \displaystyle\sum_{u \in T_i} \lvert u \rvert}$ where $\ell_i$ and $L_i$ are defined in \Cref{liLi}. For all $i \geq 1$ we define the sets ${T_{i,a} := \{u \in T_i \mid u_1 = a\}}$ and ${T_{i,b} := \{u \in T_i \mid u_1 = b\}}$, so that $T_{i,b} = \overline{T_{i,a}}$. We also define the map
	\begin{align*}
		g : T \backslash \{\varepsilon\} & \longrightarrow T \backslash \{\varepsilon\} \\
		u & \longmapsto \mathcal{P}_{f,\overline{u_1}}\left(\overline{\mathcal{D}_f(u)}\right)
	\end{align*}
	For all $u \in T \backslash \{\varepsilon\}$ we have $g(g(u)) = \mathcal{P}_{f,u_1}\left(\overline{\mathcal{D}_f\left(\mathcal{P}_{f,\overline{u_1}}\left(\overline{\mathcal{D}_f(u)}\right)\right)}\right) = \mathcal{P}_{f,u_1}\left(\overline{\overline{\mathcal{D}_f(u)}}\right) = u$ so $g$ is a bijection. Also, for all $i \geq 1$ we observe that $T_{i,b} = g(T_{i,a})$ so we deduce that $f(i) = \displaystyle \sum_{u \in T_{i,a}} \left[\lvert u \rvert + \left\lvert g(u) \right\rvert\right]$ for all $i \geq 1$. Now, for all $u \in T_{i,a}$ we have
	$$\lvert u \rvert = a \left\lvert\mathcal{D}_f(u)\right\rvert_a + b \left\lvert\mathcal{D}_f(u)\right\rvert_b + 2a$$
	and
	$$\left\lvert g(u) \right\rvert = a \left\lvert\mathcal{D}_f\left(g(u)\right)\right\rvert_a + b \left\lvert\mathcal{D}_f\left(g(u)\right)\right\rvert_b + 2a = a \left\lvert\mathcal{D}_f(u) \right\rvert_b + b \left\lvert\mathcal{D}_f(u)\right\rvert_a + 2a.$$
	Moreover, we have $\mathcal{D}_f(T_{i,a}) = T_{i-1}$ and $\lvert T_{i,a} \rvert = 2^{i-1}$ for all $i \geq 1$. We deduce that
	\begin{align*}
		f(i+1) = & \displaystyle \sum_{u \in T_{i+1,a}} \left[\lvert u \rvert + \left\lvert g(u)\right \rvert\right] = \sum_{u \in T_{i+1,a}} \left[(a+b)\left\lvert\mathcal{D}_f(u)\right\rvert + 4a\right] \\
		= &~ (a+b) \sum_{u \in T_{i+1,a}} \left\lvert\mathcal{D}_f(u)\right\rvert + \sum_{u \in T_{i+1,a}} 4a = (a+b) \sum_{u \in T_i} \lvert u\rvert + 4a 2^i = (a+b) f(i) + 4a 2^i.
	\end{align*}
	for all $i \geq 0$. Then, with $f(0) = 0$ and by setting $\psi := \frac{4a}{a+b-2}$, a quick induction on $i$ yields
	\begin{equation}\label{avTi}
		f(i) = \psi (a+b)^i - \psi 2^i
	\end{equation}
	for all $i \geq 0$.
	
	The second step is to compute $p_i(n)$ for large enough $n$. For $i \geq 0$, if $n < \ell_i$ then $b_i(n) = s_i(n) = p_i(n) = 0$, and $s_i(\ell_i) = 0$. Also, if $n > L_i$ then $b_i(n) = 0$ and $s_i(n) = \lvert T_i \rvert = 2^i$. Then, for all $n > L_i$, we get
	\begin{align*}
		p_i(n) = &~ \displaystyle \sum_{m = 0}^{\ell_i} s_i(m) + \sum_{m = \ell_i+1}^{L_i} s_i(m) + \sum_{m = L_i+1}^{n-1} s_i(m) = 0 + \sum_{m = \ell_i+1}^{L_i} \sum_{m' = \ell_i}^{m-1} b_i(m') + \sum_{m = L_i+1}^{n-1} 2^i \\
		= &~ \sum_{m = \ell_i}^{L_i} (L_i-m)b_i(m) + \left(n-L_i-1\right)2^i = L_i \sum_{m = \ell_i}^{L_i} b_i(m) - f(i) + \left(n-L_i-1\right) 2^i \\
		= &~ L_i 2^i - \psi (a+b)^i + \psi 2^i + \left(n-L_i-1\right)2^i \textup{ thanks to \Cref{avTi}},
	\end{align*}
	which yields
	\begin{equation}\label{pi}
		p_i(n) = (n+\psi-1)2^i - \psi(a+b)^i
	\end{equation}
	for all $n > L_i$.
	
	The last step is to deduce the lower bound of $p_{\mathcal{C}_f^\infty}(n)$. If we fix $i \geq 0$, we notice that the function $p_i : \mathbb{N} \rightarrow \mathbb{N}$ is convex since its second finite difference is $b_i : \mathbb{N} \rightarrow \mathbb{N}$. Also, \Cref{pi} implies that this function is eventually affine so we deduce that $p_i(n) \geq \max \left(0, \left(n+\psi-1\right)2^i - \psi(a+b)^i \right)$ for all $i \geq 0$ and all $n \geq 0$. For $n \geq 1$, the set $I_n := \left\{i \geq 0 ~\middle|~ \left(n+\psi-1\right)2^i - \psi(a+b)^i \geq 0 \right\}$ satisfies $I_n = \llbracket 0,j \rrbracket$ where $j := \left\lfloor \frac{\log(m)}{\log\left(\frac{a+b}{2}\right)} \right\rfloor$ and $m := \frac{n+\psi-1}{\psi} \geq 1$, and we get
	\begin{align*}
		p(n) = &~ \displaystyle \sum_{i=0}^\infty p_i(n) \geq \displaystyle \sum_{i \in I_n} \left[\left(n+\psi-1\right)2^i - \psi(a+b)^i\right] = \psi \displaystyle \sum_{i=0}^j \left[m2^i - (a+b)^i\right] \\
		\geq &~ \psi \left[m\left(2^{j+1}-1\right) - \frac{(a+b)^{j+1}-1}{a+b-1} \right] = \psi \left[m2^{j+1} - \frac{(a+b)^{j+1}}{a+b-1}\right] - \psi m + \frac{\psi}{a+b-1} \\
		\geq &~ \psi h(j) - n + 1 - \frac{4a}{a+b-1} \geq \psi h(j) - n - 1 \textup{ because } a+b-1 \geq 2a
	\end{align*}
	where $h(x) = m2^{x+1} - \frac{(a+b)^{x+1}}{a+b-1}$. Now it remains to do standard estimations of $h(j)$. The function ${h : \mathbb{R}_{\geq 0} \rightarrow \mathbb{R}}$ is derivable and $h'(x) = \frac{m2^{x+1}}{\log(2)} - \frac{(a+b)^{x+1}}{(a+b-1) \log(a+b)}$, we deduce that $h$ is increasing then decreasing and that ${h(j) \geq \min\left(h\left(\frac{\log(m)}{\log\left(\frac{a+b}{2}\right)}-1\right), h\left(\frac{\log(m)}{\log\left(\frac{a+b}{2}\right)}\right)\right)}$. Now we compute
	$$h\left(\frac{\log(m)}{\log\left(\frac{a+b}{2}\right)}-1\right) = m2^{\frac{\log(m)}{\log\left(\frac{a+b}{2}\right)}} - \frac{(a+b)^{\frac{\log(m)}{\log\left(\frac{a+b}{2}\right)}}}{a+b-1} = \left[1-\frac{1}{a+b-1}\right] m^\rho$$
	and
	$$h\left(\frac{\log(m)}{\log\left(\frac{a+b}{2}\right)}\right) = m2^{\frac{\log(m)}{\log\left(\frac{a+b}{2}\right)}+1} - \frac{(a+b)^{\frac{\log(m)}{\log\left(\frac{a+b}{2}\right)}+1}}{a+b-1} = 2 m^{\frac{\log(2)}{\log\left(\frac{a+b}{2}\right)}+1} - \frac{(a+b)}{a+b-1} m^\rho = \left[2-\frac{a+b}{a+b-1}\right] m^\rho = \left[1-\frac{1}{a+b-1}\right] m^\rho.$$
	We deduce that $p(n) \geq \psi\left[1-\frac{1}{a+b-1}\right] m^\rho - n - 1$ for all $n \geq 1$. If ${\psi \geq 1}$, then $m \geq \frac{n}{\psi}$ so \Cref{p3p} yields ${p_{\mathcal{C}_f^\infty}(n) \geq \psi^{1-\rho}\left[1-\frac{1}{a+b-1}\right] n^\rho}$ for all ${n \geq 1}$. If $\psi < 1$, then $m \geq n$ so \Cref{p3p} yields ${p_{\mathcal{C}_f^\infty}(n) \geq \psi\left[1-\frac{1}{a+b-1}\right] n^\rho}$ for all $n \geq 1$.
\end{proof}

\subsection{Upper bounds}

We prove here \Cref{Th3,Th4} based on \Cref{upperBound} where we show that a lower bound of $\ell_i$ provides an upper bound of $p_{\mathcal{C}_f^\infty}(n)$.

\subsubsection{An upper bound from a lower bound of $\ell_i$}

\begin{proposition}\label{upperBound}
	Let $\mathpzc{A} = \{a,b\}$ be a binary alphabet. If there exists $\lambda > 1$, $C > 0$ and $D \geq 0$ such that $\ell_i \geq C \lambda^i - D - 1$ for all $i \geq 0$, then we have $p_{\mathcal{C}_f^\infty}(n) = \mathcal{O}(n^\chi)$ where $\chi := \frac{\log(2\lambda)}{\log(\lambda)}$.
\end{proposition}
\begin{proof}
	With the notation introduced in \Cref{Tcontribution}, we have $s_i(n) = 0$ if $n \leq \ell_i$ and $s_i(n) \leq \lvert T_i \rvert = 2^i$ otherwise, therefore $p_i(n) \leq \max\left(0,2^i(n-\ell_i-1)\right)$ for all $i \geq 0$. Now, for $n \geq \max(0,C-D)$, the set $I_n := \{i \geq 0 \mid n-\ell_i-1 \geq 0\}$ satisfies $I_n \subseteq \llbracket 0,j \rrbracket$ where $j := \left\lfloor \frac{\log(m)}{\log(\lambda)} \right\rfloor$ and $m := \frac{n+D}{C} \geq 1$, and we get
	$$p(n) = \displaystyle\sum_{i=0}^\infty p_i(n) \leq \sum_{i \in I_n} \left(n-\ell_i-1\right) 2^i \leq \sum_{i =0}^j n 2^i = n2^{j+1}-n.$$
	Then, for all $n \geq D$, we have $j \leq \frac{\log(m)}{\log(\lambda)}$ and $m \leq \frac{2n}{C}$ so
	$$p(n) \leq 2n2^\frac{\log(m)}{\log(\lambda)} - n = 2nm^\frac{\log(2)}{\log(\lambda)} - n \leq 2\left(\frac{2}{C}\right)^\frac{\log(2)}{\log(\lambda)} n^\chi - n.$$
	Finally, \Cref{p3p} yields ${p_{\mathcal{C}_f^\infty}(n) \leq 6\left(\frac{2}{C}\right)^\frac{\log(2)}{\log(\lambda)} n^\chi - 2n + 1 \leq 6\left(\frac{2}{C}\right)^\frac{\log(2)}{\log(\lambda)} n^\chi}$ for all ${n \geq \max(1,C-D,D)}$.
\end{proof}

The rest of the section is therefore dedicated to finding lower bounds of $\ell_i$.

\subsubsection{A general upper bound}

We provide a first unrefined lower bound of $\ell_i$ over any binary alphabet, which yields a general upper bound of $p_{\mathcal{C}_f^\infty}(n)$ and will also be useful later to prove \Cref{Th4} (iii).

\begin{lemma}\label{roughli}
	Let $\mathpzc{A} = \{a,b\}$ be a binary alphabet and let $\lambda := \frac{a(2b+1)}{a+b+1}$. Then we have $\ell_i > \lambda^i$ for all $i \geq 1$.
\end{lemma}

\begin{proof}
	Let $i \geq 1$, let $u \in T_{i+1}$ and let $v \in T_i$ be such that $u = \mathcal{P}_{f,u_1}(v)$. We first show that
	\begin{equation}\label{roughFreq}
		\frac{\lvert v \rvert_b}{\lvert v \rvert} \geq \frac{a}{a+b+1}.
	\end{equation}
	If $i = 1$, we have $v \in T_1 = \{a^a b^a, b^a a^a\}$ so $\frac{\lvert v \rvert_b}{\lvert v \rvert} = \frac{1}{2} > \frac{a}{a+b+1}$.
	If $i \geq 2$, $v$ is canonically factorized as $a_0^a a_1^{p_1} ... a_n^{p_n} a_{n+1}^a$ where $n \geq 2a$ so there exists $\ell \geq a$,  $q_1,q_2, ..., q_\ell \in \mathpzc{A}$, $w \in \{\varepsilon, b^a \}$ and $w' \in \{\varepsilon, a^a \}$ such that
	$$v = w~a^{q_1} b^{q_2} ... a^{q_{2\ell-1}} b^{q_{2\ell}}~w'.$$
	If $(w,w') = (\varepsilon, \varepsilon)$, we have $\lvert v \rvert_a \leq b \ell$ and $\lvert v \rvert_b \geq a \ell$ so $\frac{\lvert v \rvert_b}{\lvert v \rvert} = \frac{1}{1+\frac{\lvert u \rvert_a}{\lvert v \rvert_b}} \geq \frac{a}{a+b} > \frac{a}{a+b+1}$. If $(w,w') = (\varepsilon, a^a)$, we have $\lvert v \rvert_a \leq b\ell+a$ and $\lvert v \rvert_b \geq a\ell$ so $\frac{\lvert v \rvert_b}{\lvert v \rvert} \geq \frac{a}{a+b+\frac{a}{\ell}} \geq \frac{a}{a+b+1}$. If $(w,w') = (b^a, \varepsilon)$, we have $\lvert v \rvert_a \leq b \ell$ and $\lvert v \rvert_b \geq a(\ell+1)$ so $\frac{\lvert v \rvert_b}{\lvert v \rvert} = \frac{1}{1+\frac{\lvert u \rvert_a}{\lvert v \rvert_b}} \geq \frac{a}{a+b-\frac{b}{\ell+1}} > \frac{a}{a+b+1}$. If $(w,w') = (b^a, a^a)$, we have $\lvert v \rvert_a \leq b \ell + a$ and $\lvert v \rvert_b \geq a(\ell+1)$ so $\frac{\lvert v \rvert_b}{\lvert v \rvert} \geq \frac{a}{a+b-\frac{b-a}{\ell+1}} > \frac{a}{a+b+1}$. This concludes the proof of \Cref{roughFreq}, which yields
	$$\lvert u \rvert = a \lvert v \rvert_a + b \lvert v \rvert_b + 2a \geq a\lvert v \rvert + (b-a)\lvert v \rvert_b + 2a > \left[a+(b-a)\frac{a}{a+b+1}\right] \lvert v \rvert \geq \lambda \ell_i.$$
	
	This holds for all $u \in T_{i+1}$ so we deduce that $\ell_{i+1} > \lambda \ell_i$ for all $i \geq 1$, and a quick induction on $i$ with $\ell_1 = 2a > \lambda$ yields the result.
\end{proof}

Then \Cref{upperBound} and \Cref{roughli} prove \Cref{Th3} (i).

\subsubsection{The upper bound over even alphabets}

Over even alphabets, we compute $\ell_i$ easily.

\begin{lemma}\label{evenli}
	Let $\mathpzc{A} = \{a,b\}$ be an even alphabet and let $\psi := \frac{4a}{a+b-2}$. Then we have $\ell_i = \psi \left(\frac{a+b}{2}\right)^i - \psi$ for all $i \geq 0$.
\end{lemma}
\begin{proof}
	Because $a$ and $b$ are even, a quick induction on the height of $u$ shows that every word $u \in T$ has even length and satisfies $\lvert u \rvert_a = \lvert u \rvert_b = \frac{\lvert u \rvert}{2}$. Then, with the fact that $\left\lvert \mathcal{P}_{f,c}(u) \right\rvert = a\lvert u \rvert_a + b\lvert u \rvert_b + 2a$ for all $u \in \mathpzc{A}^*$ and all $c \in \mathpzc{A}$, a quick induction on $i$ shows that the length of each word in $T_i$ is $\psi \left(\frac{a+b}{2}\right)^i - \psi$ for all $i \geq 0$. This yields $\ell_i = L_i = \psi \left(\frac{a+b}{2}\right)^i - \psi$ for all $i \geq 0$.
\end{proof}

Then \Cref{upperBound} and \Cref{evenli} prove \Cref{Th3} (ii).

\begin{remark}
	In the proof of \Cref{evenli} we show that every word in $T_i$ has the same length, i.e., $\ell_i = L_i$ for all $i \geq 0$ over even alphabets. This directly implies that \Cref{Lili1} holds over even alphabets.
\end{remark}

\subsubsection{The upper bound over odd alphabets}

Over odd alphabets, we compute the growth rate of $\ell_i$.

\newcommand{\bigzero}{\mbox{\normalfont\Large 0}}
\newcommand{\prim}{
    \left(\begin{gathered}
    \tikzpicture[every node/.style={anchor=south west}]
        \node[minimum width=1cm,minimum height=0.4cm] at (0,0.9) {$R$};
        \node[minimum width=1cm,minimum height=0.4cm] at (0.8,0.9) {$S$};
        \node[minimum width=1cm,minimum height=0.4cm] at (0,0.1) {$\bigzero$};
        \node[minimum width=1cm,minimum height=0.4cm] at (0.8,0.05) {$\frac{b-1}{2}$};
        \draw (0.9,0) -- (0.9,1.6);
        \draw (0.1,0.8) -- (1.6,0.8);
    \endtikzpicture
    \end{gathered}\right)
}
\newcommand{\primk}{
    \left(\begin{gathered}
    \tikzpicture[every node/.style={anchor=south west}]
        \node[minimum width=1cm,minimum height=0.4cm] at (0,0.9) {$R^i$};
        \node[minimum width=1cm,minimum height=0.4cm] at (0.8,0.85) {$S_i$};
        \node[minimum width=1cm,minimum height=0.4cm] at (0,0.1) {$\bigzero$};
        \node[minimum width=1cm,minimum height=0.4cm] at (0.8,0) {$\left(\frac{b-1}{2}\right)^i$};
        \draw (0.9,0) -- (0.9,1.6);
        \draw (0.1,0.8) -- (1.6,0.8);
    \endtikzpicture
    \end{gathered}\right)
}

\begin{lemma}\label{oddli}
	Let $\mathpzc{A} = \{a,b\}$ be an odd alphabet and let $\lambda := \begin{cases} \frac{1+\sqrt{2b-1}}{2} \textit{ if } a = 1, \\ \textit{the dominant root of } X^3 - \frac{a+b}{2} X^2 + \frac{(b-a)^2}{4} \textit{ otherwise.} \end{cases}$
	
	Then there exists $C > 0$ and $D \geq 0$ such that $\ell_i \geq C \lambda^i - D - 1$ for all $i \geq 0$.
\end{lemma}
\begin{proof}
	The first step is to show that every word $u \in T$ has even length. We proceed by induction on the height of $u$: if $u \in T_0 := \{\varepsilon\}$ this is trivial, now let $i \geq 0$ be such that every word of $T_i$ has even length. If $u \in T_{i+1}$, let $v \in T_i$ be such that $u = \mathcal{P}_{f,u_1}(v)$, so that $\lvert u \rvert = a \lvert v \rvert_a + b \lvert v \rvert_b + 2a$. By hypothesis, $v$ has even length so $\lvert v \rvert_a$ and $\lvert v \rvert_b$ are either both even or both odd, and in the two cases $\lvert u \rvert$ is even.
	
	The second step is to show that $\ell_i = \left\lvert \left(\mathcal{P}_{f,a}\right)^i(\varepsilon) \right\rvert$ for all $i \geq 0$. To do so, we introduce the following notation: if $c \in \mathpzc{A}$ and $u \in \mathpzc{A}^*$, then $\lvert u \rvert_{c,1}$ (resp. $\lvert u \rvert_{c,0}$) denotes the number of occurrences of the letter $c$ at odd (resp. even) indexes in $u$. For $u \in \mathpzc{A}^*$, we define
	$$V(u) := \left(\begin{matrix}
					\lvert u \rvert_{a,0} \\
					\lvert u \rvert_{a,1} \\
					\lvert u \rvert_{b,0} \\
					\lvert u \rvert_{b,1}
			  \end{matrix}\right) \textrm{,}~~
	  M := \left(\begin{matrix}
				\frac{a-1}{2} & 0 & \frac{b-1}{2} & 0 \\
				\frac{a+1}{2} & 0 & \frac{b+1}{2} & 0 \\
				0 & \frac{a+1}{2} & 0 & \frac{b+1}{2} \\
				0 & \frac{a-1}{2} & 0 & \frac{b-1}{2}
		   \end{matrix}\right) \textrm{,}~~
	  N := \left(\begin{matrix}
				\frac{a-1}{2} \\
				\frac{a+1}{2} \\
				\frac{a+1}{2} \\
				\frac{a-1}{2}
		   \end{matrix}\right) ~\textrm{and}~~
	  P := \left(\begin{matrix}
				0 & 0 & 1 & 0 \\
				0 & 0 & 0 & 1 \\
				1 & 0 & 0 & 0 \\
				0 & 1 & 0 & 0
		   \end{matrix}\right)\textrm{,}$$
	so that
	\begin{equation}\label{matrixPa}
		V\left(\mathcal{P}_{f,a}(u)\right) = M~V(u) + N
	\end{equation}
	and
	\begin{equation}\label{matrixPb}
		V\left(\mathcal{P}_{f,b}(u)\right) = P~V\left(\mathcal{P}_{f,a}(u)\right) 
	\end{equation}
	for all $u \in \mathpzc{A}^*$ of even length. We also define $\ell_i' := \displaystyle\min_{u \in T_i} \lvert u \rvert_b$ and $\ell''_i := \displaystyle\min_{u \in T_i} \lvert u \rvert_{b,1}$ for all $i \geq 0$. Note that, since every word of $T$ has even length and each $T_i$ is reversal-invariant, we have $\ell''_i = \displaystyle\min_{u \in T_i} \lvert u \rvert_{b,0}$. Now we show by induction on $i$ that $\ell_i = \left\lvert \left(\mathcal{P}_{f,a}\right)^i(\varepsilon) \right\rvert$, $\ell'_i = \left\lvert \left(\mathcal{P}_{f,a}\right)^i(\varepsilon) \right\rvert_b$ and $\ell''_i = \left\lvert \left(\mathcal{P}_{f,a}\right)^i(\varepsilon) \right\rvert_{b,1}$ for all $i \geq 0$. For $i = 0$ this is trivial, now let $i \geq 0$ be such that the three equalities hold. Let $u \in T_{i+1}$ and let $v \in T_i$ be such that $u = \mathcal{P}_{f,u_1}(v)$. If $u_1 = a$, by noticing that $\left\lvert v \right\rvert_{a,1} + \left\lvert v \right\rvert_{b,1} = \frac{\lvert v \rvert}{2}$, \Cref{matrixPa} yields
	$$\lvert u \rvert_{b,1} = \frac{a-1}{2}\frac{\lvert v \rvert}{2} + \frac{b-a}{2} \lvert v \rvert_{b,1} + \frac{a-1}{2} \geq \frac{a-1}{2}\frac{\ell_i}{2} + \frac{b-a}{2} \ell''_i + \frac{a-1}{2} = \left\lvert \left(\mathcal{P}_{f,a}\right)^{i+1}(\varepsilon) \right\rvert_{b,1},$$
	$$\lvert u \rvert_b = a\frac{\lvert v \rvert}{2} + (b-a) \lvert v \rvert_{b,1} + a \geq a\frac{\ell_i}{2} + (b-a) \ell''_i + a = \left\lvert \left(\mathcal{P}_{f,a}\right)^{i+1}(\varepsilon) \right\rvert_b,$$
	$$\lvert u \rvert = a \lvert v \rvert + (b-a) \lvert v \rvert_b + 2a \geq a \ell_i + (b-a) \ell'_i + 2a = \left\lvert \left(\mathcal{P}_{f,a}\right)^{i+1}(\varepsilon) \right\rvert.$$
	If $u_1 = b$, then \Cref{matrixPb} yields
	$$\lvert u \rvert_{b,1} = \frac{a+1}{2}\frac{\lvert v \rvert}{2} + \frac{b-a}{2} \lvert v \rvert_{b,0} + \frac{a+1}{2} \geq \frac{a-1}{2}\frac{\ell_i}{2} + \frac{b-a}{2} \ell''_i + \frac{a-1}{2} = \left\lvert \left(\mathcal{P}_{f,a}\right)^{i+1}(\varepsilon) \right\rvert_{b,1},$$
	$$\lvert u \rvert_b = a\frac{\lvert v \rvert}{2} + (b-a) \lvert v \rvert_{b,0} + a \geq a \frac{\ell_i}{2} + (b-a) \ell''_i + a = \left\lvert \left(\mathcal{P}_{f,a}\right)^{i+1}(\varepsilon) \right\rvert_b,$$
	$$\lvert u \rvert = a\lvert v \rvert + (b-a) \lvert v \rvert_b + 2a \geq a \ell_i + (b-a) \ell'_i + 2a = \left\lvert \left(\mathcal{P}_{f,a}\right)^{i+1}(\varepsilon) \right\rvert.$$
	
	The last step is to deduce the lower bound of $\ell_i$. First, \Cref{matrixPa} yields
	$${V\left(\left(\mathcal{P}_{f,a}\right)^{i+1}(u)\right) = M~V\left(\left(\mathcal{P}_{f,a}\right)^i(u)\right) + N}$$ for all $i \geq 0$, and a quick induction on $i$ yields $V\left(\left(\mathcal{P}_{f,a}\right)^i(u)\right) = \displaystyle\sum_{j=0}^{i-1} M^j N$ for all $i \geq 0$. For all $i \geq 1$, we deduce that
	\begin{equation}\label{n1}
		\ell_i = \left\lvert \left(\mathcal{P}_{f,a}\right)^i(\varepsilon) \right\rvert = \left\lVert V\left(\left(\mathcal{P}_{f,a}\right)^i(u)\right) \right\rVert_1 \geq \left\lVert M^{i-1}N \right\rVert_1
	\end{equation}
	
	If $a = 1$, we have $M = \prim$ where $R := \left(\begin{matrix}
					0 & 0 & \frac{b-1}{2} \\
					1 & 0 & \frac{b+1}{2} \\
					0 & 1 & 0
			\end{matrix}\right)$. Then, for all $i \geq 0$ there exists $S_i \in \mathbb{R}^{1 \times 3}$ such that $M^i = \primk$ and $M^i N = \left\lVert R^i \left(\begin{matrix} 0 \\ 1 \\ 1 \end{matrix}\right) \right\rVert_1$. One can check that $R^5$ has positive entries so $R$ is primitive and Perron-Frobenius theorem states that every entry of $R^i$ grows like $\Theta(\lambda^i)$ where $\lambda \in \mathbb{R}_+$ is the spectral radius of $R$. Moreover, the characteristic polynomial of $R$ is $(X+1)\left(X^2-X-\frac{b-1}{2}\right)$ so we deduce that $\lambda = \frac{1+\sqrt{2b-1}}{2}$. With \Cref{n1}, this directly provides $C > 0$ and $D \geq 0$ such that $\ell_i \geq C\lambda^i - D - 1$ for all $i \geq 0$.
	
	If $a > 1$, one can check that $M^2$ has positive entries so $M$ is primitive and Perron-Frobenius theorem states that every entry of $M^i$ grows like $\Theta(\lambda^i)$ where $\lambda \in \mathbb{R}_+$ is the spectral radius of $M$. Moreover, the characteristic polynomial of $M$ is $(X+1)\left(X^3 - \frac{a+b}{2} X^2 + \frac{(b-a)^2}{4}\right)$ so $\lambda$ is the dominant root of $X^3 - \frac{a+b}{2} X^2 + \frac{(b-a)^2}{4}$. With \Cref{n1}, this directly provides $C > 0$ and $D \geq 0$ such that $\ell_i \geq C\lambda^i - D - 1$ for all $i \geq 0$.
\end{proof}

Then \Cref{upperBound} and \Cref{oddli} prove \Cref{Th3} (iii).
\newpage
\begin{remark}\label{unsufficient}
	In the proof of \Cref{oddli} we can show that, more than a lower bound, $\ell_i$ grows exactly like $\lambda^i$ over odd alphabets, which means that \Cref{upperBound} cannot give a better upper bound than $\mathcal{O}(n^\xi)$. Moreover, one can prove that $L_i$ grows like $r^i$ where $r$ is the spectral radius of $MPM$ where $M$ and $P$ are defined in the last proof, and that $r > \lambda$. All this implies that \Cref{Lili1} is false over any odd alphabet. For example, over $\{1,3\}$ we have $L_i > \ell_{i+1}$ for all $i \geq 5$, starting with $L_5 = 86 > \ell_6 = 64$.
\end{remark}

\subsubsection{Upper bounds from letter frequencies}

Over any binary alphabet, we show that the letter frequencies in $\mathcal{C}_f^\infty$ provide a lower bound of $\ell_i$ and therefore an upper bound of $p_{\mathcal{C}_f^\infty}(n)$.

\begin{lemma}\label{liFreq1}
	Let $\mathpzc{A} = \{a,b\}$ be a binary alphabet. If $\displaystyle \liminf_{u \in \mathcal{C}_f^\infty} \frac{\lvert u \rvert_b}{\lvert u \rvert} > \frac{1}{2}-\phi$ for some $\phi \in \left(0,\frac{1}{2}\right)$, then there exists $C > 0$ and $D \geq 0$ such that $\ell_i \geq C \left[\frac{a+b}{2}-(b-a)\phi\right]^i - D - 1$ for all $i \geq 0$.
\end{lemma}

\begin{proof}
	By assumption, there exists $L \geq 1$ such that $\frac{\lvert u \rvert_b}{\lvert u \rvert} > \frac{1}{2}-\phi$ for all $u \in \mathcal{C}_f^\infty$ such that $\lvert u \rvert \geq L$. Let $i_0$ be the smallest integer such that $\ell_{i_0} \geq L$. Now let $i \geq i_0$, let $u \in T_{i+1}$ and let $v \in T_i$ be such that $u = \mathcal{P}_{f,u_1}(v)$. We have $\lvert v \rvert \geq \ell_i \geq \ell_{i_0} \geq L$ so $\frac{\lvert v \rvert_b}{\lvert v \rvert} > \frac{1}{2}-\phi$, so we get
	$$\lvert u \rvert = a \lvert v \rvert_a + b \lvert v \rvert_b + 2a > a \lvert v \rvert + (b-a) \lvert v \rvert_b > \left[a + (b-a)\left(\frac{1}{2}-\phi\right)\right] \lvert v \rvert \geq \left[\frac{a+b}{2} - (b-a)\phi\right] \ell_i.$$
	This holds for all $u \in T_{i+1}$ so we deduce that $\ell_{i+1} > \left[\frac{a+b}{2} - (b-a)\phi\right] \ell_i$ and $\ell_i > \left[\frac{a+b}{2} - (b-a)\phi\right]^{i-i_0} \ell_{i_0}$ for all $i \geq i_0$.
\end{proof}

Then \Cref{upperBound} and \Cref{liFreq1} prove \Cref{Th4} (i). Now we strengthen the assumption on the letter frequencies.

\begin{lemma}\label{liFreq2}
	Let $\mathpzc{A} = \{a,b\}$ be a binary alphabet. If there exists $\phi > 0$ such that $\left\lvert \frac{\lvert u \rvert_b}{\lvert u \rvert} - \frac{1}{2} \right\rvert = \mathcal{O}\left(\log(\lvert u \rvert)^{-(1+\phi)}\right)$ for all $u \in \mathcal{C}_f^\infty$, then there exists $C > 0$ and $D \geq 0$ such that $\ell_i \geq C \left(\frac{a+b}{2}\right)^i - D - 1$ for all $i \geq 0$.
\end{lemma}
\begin{proof}
	By assumption, there exists a constant $K > 0$ such that $\left\lvert\frac{\lvert u \rvert_b}{\lvert u \rvert} - \frac{1}{2}\right\rvert \leq K \log(\lvert u \rvert)^{-(1+\phi)}$ for all $u \in \mathcal{C}_f^\infty$ such that $\lvert u \rvert \geq 2$. Then, for all $i \geq 1$ and all $u \in T_i$, \Cref{roughli} yields
	\begin{equation}\label{eq}
		-\left\lvert\frac{\lvert u \rvert_b}{\lvert u \rvert} - \frac{1}{2}\right\rvert \geq -K \log(\lvert u \rvert)^{-(1+\phi)} \geq -K \log(\ell_i)^{-(1+\phi)} > -K \log(\lambda)^{-(1+\phi)} i^{-(1+\phi)}
	\end{equation}
	where $\lambda := \frac{a(2b+1)}{a+b+1} > 1$. Now let $i \geq 1$, let $u \in T_{i+1}$ and let $v \in T_i$ be such that $u = \mathcal{P}_{f,u_1}(v)$. We have
	\begin{align*}
		\lvert u \rvert &~= a \lvert v \rvert_a + b \lvert v \rvert_b + 2a = \left[\frac{a+b}{2} + (b-a) \left[\frac{\lvert v \rvert_b}{\lvert v \rvert} - \frac{1}{2}\right] \right] \lvert v \rvert + 2a \geq \left[\frac{a+b}{2} - (b-a) \left\lvert\frac{\lvert v \rvert_b}{\lvert v \rvert} - \frac{1}{2}\right\rvert \right] \lvert v \rvert
	\end{align*}
	so
	$$\log(\lvert u \rvert) \geq \log(\lvert v \rvert) + \log\left(\frac{a+b}{2}\right) + \log\left(1+f(v)\right) \geq \log(\ell_i) + \log\left(\frac{a+b}{2}\right) + \log\left(1+f(v)\right)$$
	where $f(v) := -\frac{2(b-a)}{a+b} \left\lvert\frac{\lvert v \rvert_b}{\lvert v \rvert} - \frac{1}{2}\right\rvert$. Since $\left\lvert\frac{\lvert v \rvert_b}{\lvert v \rvert} - \frac{1}{2}\right\rvert \in \left[0, \frac{1}{2}\right]$, we have $f(v) \in \left[-\frac{b-a}{a+b}, 0\right]$ and ${1+f(v) \in \left[\frac{2a}{a+b}, 1\right] \subset (0,1]}$ so we get
	$$\log\left(1+f(v)\right) = - \log\left(1-\frac{f(v)}{1+f(v)}\right) \geq \frac{f(v)}{1+f(v)} \geq f(v) > -\frac{2(b-a)}{a+b} K \log(\lambda)^{-(1+\phi)} i^{-(1+\phi)} \textup{ thanks to (\ref{eq}).}$$
	This holds for every $u \in T_{i+1}$ so we deduce that
	$$\log(\ell_{i+1}) > \log(\ell_i) + \log\left(\frac{a+b}{2}\right) -\frac{2(b-a)}{a+b} K \log(\lambda)^{-(1+\phi)} i^{-(1+\phi)}$$
	and by summation
	$$\log(\ell_i) \geq \log(\ell_1) + (i-1)\log\left(\frac{a+b}{2}\right) - \frac{2(b-a)}{a+b} K \log(\lambda)^{-(1+\phi)} \sum_{j=1}^{i-1} j^{-(1+\phi)}$$
	for all $i \geq 1$. Finally, the sum $\displaystyle\sum_{j=1}^\infty j^{-(1+\phi)}$ is finite, which concludes the proof.
\end{proof}

Then \Cref{upperBound} and \Cref{liFreq2} prove \Cref{Th4} (iii).

\section{Consequences over $\mathpzc{A} = \{1,2\}$}\label{consequences}

In this section we focus on the alphabet $\{1,2\}$ to study the Oldenburger-Kolakoski sequence $\kappa$. After applying our results to get precise bounds of $p_{\mathcal{C}_f^\infty}(n)$, we provide a dichotomy that links the conjectures that $\mathscr{L}(\kappa) = \mathcal{C}_f^\infty$ and that $\kappa$ is not morphic, this is \Cref{Th4}.

\subsection{The complexity of smooth words over $\{1,2\}$}\label{pCf12}

In \Cref{Th2} we have proven the conjectured lower bound $p_{\mathcal{C}_f^\infty}(n) = \Omega(n^\rho)$ where ${\rho := \frac{\log(3)}{\log(3/2)} \approx 2.70951}$. In order to compute an upper bound of $p_{\mathcal{C}_f^\infty}(n)$, we recall the best known lower bound of letter frequencies in $\mathcal{C}_f^\infty$ provided in \cite{Rao,Nilsson}, improving a result of \cite{Chvatal}.

\begin{proposition}[\cite{Rao,Nilsson}]\label{freq12}
	Over $\mathpzc{A} = \{1,2\}$, we have $\displaystyle\liminf_{u \in \mathcal{C}_f^\infty} \frac{\lvert u \rvert_2}{\lvert u \rvert} > \frac{1}{2} - 0.00008$.
\end{proposition}

We now compute the different upper bounds of $p_{\mathcal{C}_f^\infty}(n)$ over $\mathpzc{A} = \{1,2\}$, from worst to best:
\begin{itemize}
	\item $\beta \approx 7.22826$ using \Cref{firstPoly} (ii),
	
	\item $\zeta \approx 4.10628$ using \Cref{Th3} (i),
	
	\item $\gamma \approx 2.70987 < \rho + 0.00036$ using \Cref{freqCf} (ii) and \Cref{freq12},
	
	\item $\vartheta \approx 2.70974 < \rho + 0.00023$ using \Cref{Th4} (i) and \Cref{freq12}.
\end{itemize}

The upper bound $\mathcal{O}(n^\vartheta)$ completes \Cref{freqMixed12}.

\subsection{A dichotomy about $\kappa$}

We show here \Cref{dichotomy}. The first part will rely on \Cref{freqMixed12}, and the second part will rely on the following proposition and lemma.

\begin{proposition}[{\cite[Propositions 1 and 2]{Dekking81}}]\label{caracCfk}
	Let $\mathpzc{A} = \{1,2\}$.
	
	\textup{(i)} If $\mathscr{L}(\kappa)$ is complement-invariant, then $\kappa$ is recurrent.

	\textup{(ii)} $\mathscr{L}(\kappa)$ is complement-invariant if and only if $\mathscr{L}(\kappa) = \mathcal{C}_f^\infty$.
\end{proposition}

\begin{lemma}\label{synchronicity}
	If $\mathscr{L}(\kappa)$ is not complement-invariant, then there exists a word $u \in \mathscr{L}(\kappa)$ that occurs only at positions of the same parity in $S(\kappa) := \kappa_1 \kappa_2 \kappa_3 ...$
\end{lemma}
\begin{proof}
	Let $v \in \mathscr{L}(\kappa)$ be a word of minimal length such that $\overline{v} \notin \mathscr{L}(\kappa)$, and suppose that $u := \mathcal{D}_f(v)$ occurs both at an even and at an odd position in $S(\kappa)$. Let ${a_1}^{p_1} ... {a_n}^{p_n}$ be the canonical factorization of $v$ with $n \geq 0$. Since $v$ cannot be the empty word, we have $\lvert u \rvert < \lvert v \rvert$ so the minimal length of $v$ yields $u, \overline{u} \in \mathscr{L}(\kappa)$.
	
	If $(p_1, p_n) = (1,1)$, we have $u = p_2 ... p_{n-1}$. Since $u$ occurs at an even position in $S(\kappa)$, there exists $k \geq 0$ such that
	$$\kappa = \kappa_0 ... \kappa_{2k} p_2 ... p_{n-1} \kappa_{2k+n-1} ...$$
	Then the canonical factorization $b_1 {b_2}^{p_2} ... {b_{n-1}}^{p_{n-1}} b_n$ where $b_1 = 2$ is a factor of
	$${b_1}^{\kappa_{2k}} {b_2}^{p_2} ... {b_{n-1}}^{p_{n-1}} {b_n}^{\kappa_{2k+n-1}} \sqsubset  2^{\kappa_0} 1^{\kappa_1} ... 1^{\kappa_{2k-1}} 2^{\kappa_{2k}} 1^{p_2} 2^{p_3} ... = \kappa.$$
	Similarly, since $u$ occurs at an odd position in $S(\kappa)$, there exists $k \geq 0$ such that
	$$\kappa = \kappa_0 ... \kappa_{2k+1} p_2 ... p_{n-1} \kappa_{2k+n} ...$$
	Then the canonical factorization $c_1 {c_2}^{p_2} ... {c_{n-1}}^{p_{n-1}} c_n$ where $c_1 = 1$ is a factor of
	$${c_1}^{\kappa_{2k+1}} {c_2}^{p_2} ... {c_{n-1}}^{p_{n-1}} {c_n}^{\kappa_{2k+n}} \sqsubset  2^{\kappa_0} 1^{\kappa_1} ... 2^{\kappa_{2k}} 1^{\kappa_{2k+1}} 2^{p_2} 1^{p_3} ... = \kappa.$$
	By observing that $\{b_1 {b_2}^{p_2} ... {b_{n-1}}^{p_{n-1}} b_n, c_1 {c_2}^{p_2} ... {c_{n-1}}^{p_{n-1}} c_n\} = \{v, \overline{v}\}$, we deduce that $\overline{v} \sqsubset \kappa$ which is a contradiction.
	
	The proof is similar when $(p_1,p_n) = (1,2)$, $(2,1)$ or $(2,2)$.
\end{proof}

Now we can show \Cref{dichotomy}.

\begin{proof}[Proof of \Cref*{dichotomy}]
	(i) If $\mathscr{L}(\kappa) = \mathcal{C}_f^\infty$, then it is complement-invariant so \Cref{caracCfk} (i) states that $\kappa$ is recurrent. Also, we have $p_{\kappa}(n) = p_{\mathcal{C}_f^\infty}(n)$ so \Cref{freqMixed12} yields $p_{\kappa}(n) = \Omega(n^\rho)$ and $p_{\kappa}(n) = \mathcal{O}(n^{\rho+0.00023})$ where $\rho \approx 2.70951$. Since the factor complexity of any linearly recurrent sequence (resp. any morphic sequence) is $\mathcal{O}(n)$ (resp. $\mathcal{O}(n^2)$) (see \cite[Theorem 24 \textit{i)}]{DHS} (resp. \cite{Devyatov})), $\kappa$ cannot be linearly recurrent nor morphic.
	
	(ii) In \cite{CKL} the authors point out that $\kappa$ is a PD0L sequence with the two substitutions $\sigma_0 : 1 \mapsto 2, 2 \mapsto 22$ and $\sigma_1 : 1 \mapsto 1, 2 \mapsto 11$, i.e., it is a fixed point of the map that applies alternatively $\sigma_0$ and $\sigma_1$ to each letter. Now suppose that $\mathscr{L}(\kappa) \subsetneq \mathcal{C}_f^\infty$. Then \Cref{caracCfk} (ii) implies that $\mathscr{L}(\kappa)$ is not complement invariant, so \Cref{synchronicity} and the uniform recurrence of $\kappa$ provide a word $u \in \mathscr{L}(\kappa)$ that occurs with bounded gaps in $S(\kappa)$ and only at positions of the same parity. Now let $n$ be an even integer such that every word of $\mathscr{L}(\kappa) \cap \mathpzc{A}^n$ contains $u$ at least twice. Then each $v \in \mathscr{L}(\kappa) \cap \mathpzc{A}^n$ occurs only at positions of the same parity in $\kappa$, and we define the map
	$$f(v) := \begin{cases} \sigma_0(v_1) \sigma_1(v_2) ... \sigma_0(v_{n-1}) \sigma_1(v_n) \textup{ if } v \textup{ occurs only at even positions in } \kappa, \\ \sigma_1(v_1) \sigma_0(v_2) ... \sigma_1(v_{n-1}) \sigma_0(v_n) \textup{ if } v \textup{ occurs only at odd positions in } \kappa. \end{cases}$$
	Over the alphabet $\mathpzc{B} := \left\{\langle v \rangle \mid v \in \mathscr{L}(\kappa) \cap \mathpzc{A}^n\right\}$, define the substitution $\sigma$ and the coding $\tau$ by
	\begin{align*}
		\sigma :~~~ \mathpzc{B} & \longrightarrow \mathpzc{B}^+ & \tau :~~~ \mathpzc{B} & \longrightarrow \mathpzc{A} \\
		\langle v \rangle & \longmapsto \begin{cases} \langle f(v)_{[0,n)} \rangle ~~~~~~~~~~~~~~~~~~~~~\textup{ if } v_1 = 1 \\ \langle f(v)_{[0,n)} \rangle~\langle f(v)_{[1,n+1)} \rangle ~\textup{ if } v_1 = 2 \end{cases} & \langle v \rangle & \longmapsto v_1
	\end{align*}
	Then by construction we have $\kappa = \tau\left(\sigma^\infty(\langle \kappa_{[0,n)} \rangle)\right)$. Now the following statements apply to $\kappa$: a uniformly recurrent morphic sequence is linearly recurrent (\cite[Theorem 5.5.13]{BDP}), a linearly recurrent sequence has at most linear factor complexity (\cite[Theorem 24 \textit{i)}]{DHS}) and an aperiodic sequence has at least linear factor complexity (\cite{MH}). We deduce that $\kappa$ is a linearly recurrent morphic sequence and that $p_\kappa(n) = \Theta(n)$.
\end{proof}

\section*{Acknowledgements}

This work was supported by ANR-22-CE40-0011 project Inside Zero Entropy Systems.

\bibliographystyle{plainurl}
\bibliography{refs}

@article{OEIS,
	title = {The On-Line Encyclopedia of Integer Sequences},
	author = {OEIS Foundation Inc.},
	URL = {https://oeis.org}
}

@book{BDP,
	title = {Symbolic dynamics and substitutions},
	author = {M.P. B{\'e}al and F. Durand and D. Perrin},
	year = {2026},
	note = {Preprint}
}

@article{JMPR,
	title = {Frequency of patterns in smooth sequences over the alphabet $\{1,3\}$}, 
	author = {D. Jamet and I. Marcovici and L. Poirier and T. de la Rue},
	year = {2026},
	eprint = {2604.11387},
	archivePrefix = {arXiv},
	primaryClass = {math.DS}
}

@article{Devyatov,
	title = {On factor complexity of morphic sequences},
	author = {R. Devyatov},
	year = {2018},
	pages = {211-303},
	volume = {18},
	journal = {Moscow Mathematical Journal},
	DOI = {10.17323/1609-4514-2018-18-2-211-303}
}

@article{Nilsson,
	title = {Letter frequencies in the \uppercase{K}olakoski sequence},
	author = {J. Nilsson},
	year = {2014},
	pages = {549-552},
	volume = {126},
	journal = {Acta \uppercase{P}hysica \uppercase{P}olonica \uppercase{A}},
	DOI = {10.12693/APhysPolA.126.549}
}

@article{Rao,
	title = {Trucs et bidules sur la séquence de \uppercase{K}olakoski},
	author = {M. Rao},
	year = {2012},
	URL = {https://www.arthy.org/kola/kola.php}
}

@article{Huang,
	title = {The complexity of smooth words on 2-letter alphabets},
	author = {Y. B. Huang},
	year = {2011},
	journal = {Theoretical Computer Science},
	volume = {412},
	number = {45},
	pages = {6327-6339},
	ISSN = {0304-3975},
	DOI = {https://doi.org/10.1016/j.tcs.2011.07.002}
}

@article{powerFree,
	title = {The powers of smooth words over arbitrary 2-letter alphabets}, 
	author = {Y. B. Huang},
	year = {2011},
	eprint = {0904.0562},
	archivePrefix = {arXiv},
	primaryClass = {math.CO}
}

@article{Sing10,
	title = {More {K}olakoski sequences},
	author = {B. Sing},
	year = {2011},
	journal = {Integers},
	fjournal = {Integers. Electronic Journal of Combinatorial Number Theory},
	volume = {11B},
	pages = {Paper No. A14, 17},
	issn = {1553-1732},
	eprint = {1009.4061},
	archivePrefix = {arXiv}
}

@article{HW,
	title = {A note on the complexity of $\uppercase{C}^\infty$-words},
	author = {Y. B. Huang and W. D. Weakley},
	year = {2010},
	journal = {Theor. Comput. Sci.},
	volume = {411},
	number = {40-42},
	pages = {3731-3735},
	DOI = {10.1016/J.TCS.2010.06.024}
}

@article{BJP,
	title = {Smooth words on 2-letter alphabets having same parity},
	author = {S. Brlek and D. Jamet and G. Paquin},
	year = {2008},
	journal = {Theoretical Computer Science},
	volume = {393},
	number = {1},
	pages = {166-181},
	ISSN = {0304-3975},
	DOI = {https://doi.org/10.1016/j.tcs.2007.11.019}
}

@article{BMP,
	title = {Properties of Extremal Infinite Smooth Words},
	author = {S. Brlek and G. Melançon and G. Paquin},
	year = {2007},
	pages = {},
	volume = {9},
	journal = {Discrete Mathematics and Theoretical Computer Science},
	DOI = {10.46298/dmtcs.412}
}

@article{BS,
	title = {Kolakoski-(3,1) is a (deformed) model set},
	author = {M. Baake and B. Sing},
	year = {2004},
	volume = {47},
	number = {2},
	journal = {Canadian Mathematical Bulletin},
	pages = {168–190},
	DOI={10.4153/CMB-2004-018-6}
}

@article{Sing02,
	title = {Spektrale eigenschaften der \uppercase{K}olakoski-sequenzen},
	author = {B. Sing},
	year = {2002},
	journal = {Diploma thesis, Universität Tübingen},
	url = {https://www.bb-math.com/bernd/pub.html#kolakoski}
}

@article{Dekking01,
	title = {What is the long range order in the \uppercase{K}olakoski sequence?},
	author = {F. M. Dekking},
	year = {2001},
	pages = {},
	isbn = {978-90-481-4832-5},
	DOI = {10.1007/978-94-015-8784-6_5}
}

@article{DHS,
	author = {F. Durand and B. Host and C. Skau},
	title = {Substitutional dynamical systems, {B}ratteli diagrams and dimension groups},
	year = {1999},
	journal = {Ergodic Theory Dynam. Systems},
	fjournal = {Ergodic Theory and Dynamical Systems},
	volume = {19},
	number = {4},
	pages = {953-993},
	issn = {0143-3857,1469-4417},
	mrclass = {46L55 (37B99 46L80 54H20)},
	mrnumber = {1709427},
	mrreviewer = {A.\ I.\ Danilenko},
	doi = {10.1017/S0143385799133947},
	url = {https://doi.org/10.1017/S0143385799133947}
}

@article{Cassaigne,
	title = {Complexité et facteurs spéciaux},
	author = {J. Cassaigne},
	year = {1997},
	pages = {},
	volume = {4},
	journal = {Bulletin of the Belgian Mathematical Society - Simon Stevin},
	DOI = {10.36045/bbms/1105730624}
}

@article{Carpi,
	title = {On repeated factors in $\uppercase{C}^\infty$-words},
	author = {A. Carpi},
	year = {1994},
	journal = {Information Processing Letters},
	volume = {52},
	number = {6},
	pages = {289-294},
	ISSN = {0020-0190},
	DOI = {10.1016/0020-0190(94)00162-6}
}

@article{Chvatal,
	title = {Notes on the \uppercase{K}olakoski sequence.},
	author = {V. Chvátal},
	year = {1993},
	journal = {Technical Report 93-84},
	publisher = {DIMACS}
}

@inbook{CKL,
	title = {Alternating iteration of morphisms and Kolakovski sequence},
	author = {K. \v{C}ulik II and J. Karhumäki and A. Lepistö},
	year = {1992},
	editors = {G. Rozenberg and A. Salomaa},
	booktitle = {Lindermayer systems: impacts on theoretical computer science, computer graphics, and developmental biology},
	pages = {93-106},
	publisher = {Springer Berlin Heidelberg},
	doi = {10.1007/978-3-642-58117-5_4}
}

@inproceedings{Keane,
	title = {Ergodic theory and subshifts of finite type},
	author = {M. S. Keane},
	year = {1991},
	booktitle = {Ergodic theory, symbolic dynamics, and hyperbolic spaces.},
	ISBN = {0-19-853390-X; 0-19-859685-5},
	pages = {35-70},
	publisher = {Oxford Univ. Press, New York}
}

@article{Weakley,
	title = {On the number of $\uppercase{C}^\infty$-words of each length},
	author = {W. D. Weakley},
	year = {1989},
	journal = {J. Comb. Theory {A}},
	volume = {51},
	number = {1},
	pages = {55-62},
	DOI = {10.1016/0097-3165(89)90076-9}
}

@article{Dekking81,
	title = {On the structure of self-generating sequences},
	author = {F. M. Dekking},
	year = {1981},
	ISSN = {09895558},
	journal = {Séminaire de théorie des nombres de Bordeaux},
	pages = {1-6},
	publisher = {Société arithmétique de Bordeaux},
	URL = {http://www.jstor.org/stable/44166389}
}

@article{Dekking79,
	title = {Regularity and irregularity of sequences generated by automata},
	author = {F. M. Dekking},
	year = {1979},
	journal = {Seminaire de Théorie des Nombres de Bordeaux},
	language = {eng},
	pages = {1-10},
	volume = {9},
	URL = {http://eudml.org/doc/182065}
}

@article{Ucoluk,
	title = {Solution of Advanced Problem 5304},
	author = {W. Kolakoski and N. Üçoluk},
	year = {1966},
	journal = {The American Mathematical Monthly},
	volume = {73},
	pages = {681-682}
}

@article{Kolakoski,
	title = {Self generating runs, Problem 5304},
	author = {W. Kolakoski},
	year = {1965},
	journal = {The American Mathematical Monthly},
	volume = {72},
	pages = {674},
	DOI = {10.2307/2313883}
}

@article{Oldenburger,
	title = {Exponent Trajectories in Symbolic Dynamics},
	author = {R. Oldenburger},
	year = {1939},
	ISSN = {00029947},
	journal = {Transactions of the American Mathematical Society},
	number = {3},
	pages = {453-466},
	publisher = {American Mathematical Society},
	urldate = {2024-02-12},
	volume = {46},
	DOI = {10/2307/1989933}
}

@article{MH,
	title = {Symbolic dynamics},
	author = {M. Morse and G. A. Hedlund},
	year = {1938},
	journal = {American journal of mathematics},
	number = {3},
	pages = {815–866},
	volume = {60},
	doi = {https://doi.org/10.2307/2371264}
}

\end{document}